\documentclass[10pt,letterpaper]{article} 
\usepackage[margin=1in]{geometry}

\usepackage[utf8]{inputenc} 
\usepackage[T1]{fontenc}    
\usepackage{hyperref}       
\usepackage{url}            
\usepackage{booktabs}       
\usepackage{amsfonts}       
\usepackage{nicefrac}       
\usepackage{microtype}      
\usepackage{xcolor}         
\usepackage{xspace}

\usepackage{amsthm}
\usepackage{amssymb}
\usepackage{amsmath}
\usepackage{xcolor}
\usepackage{algorithm}
\usepackage{algorithmic}
\usepackage{float}
\usepackage{natbib}
\usepackage{graphicx}
\usepackage{cleveref}

\newcommand{\eps}{\varepsilon}

\newcommand{\po}{p_{f}}
\newcommand{\pl}{p_{f}}

\newcommand{\E}{\mathbb{E}}
\newcommand{\bern}[1]{\mathrm{Bernoulli}(#1)}
\newcommand{\bin}[2]{\mathrm{B}({#1},{#2})}
\newcommand{\poly}{\operatorname{poly}}

\newcommand{\edgeflip}{\mathbb{M}}

\newcommand{\alg}{\mathcal{A}}

\newtheorem{theorem}{Theorem}
\newtheorem{definition}[theorem]{Definition}
\newtheorem{lemma}[theorem]{Lemma}

\newtheorem{claim}[theorem]{Claim}

\newtheorem{fact}[theorem]{Fact}

\newcommand{\defn}[1]{\textbf{\emph{#1}}}

\newif\ifappendix
\appendixtrue

\newcommand{\refappendix}[1]{\ifappendix#1\xspace
\else
 the supplementary material\xspace
\fi}

\begin{document}

\title{An Efficient Private Algorithm for Community Detection}

\author{\large Vincent Cohen-Addad \\ \normalsize Google Research \\ \small\texttt{cohenaddad@google.com}
  \and
\large Alessandro Epasto \\ \normalsize Google Research \\ \small\texttt{aepasto@google.com}
\and
\large Haim Kaplam \\ \normalsize Tel Aviv University, Google Research \\ \small\texttt{haimk@google.com}
\and
\large Hanna Koml\'os \\ \normalsize Max Planck Institute for Informatics \\ \small\texttt{hkomlos@gmail.com}
\and
\large Silvio Lattanzi \\ \normalsize Google Research \\ \small\texttt{silviol@google.com}
}

\date{}

\maketitle

\begin{abstract}
    In this paper, we study the community detection problem in the stochastic block model (SBM) under privacy constraints. We introduce private and  highly efficient algorithms for exact community detection within the SBM framework. Our algorithms represent the first differentially private methods capable of achieving exact recovery in a wide range of model parameters with near-linear time and space complexity. This is a significant improvement over previous SBM recovery algorithms, which either required pseudo-polynomial time or a quadratic scaling of resources for a constant privacy budget.

    Central to our approach is the introduction of a new concept, \emph{adaptive disjoint-star algorithms}. These algorithms efficiently explore the graph's structure by querying node degrees on edge-disjoint subgraphs. We demonstrate that this general class of algorithms inherently offers strong privacy guarantees, a result that potentially holds value beyond the scope of SBM community detection. Finally, in we perform an empirical analysis of our algorithms showing that they can scale exact recovery on graphs with two orders of magnitude more nodes than prior work.
\end{abstract}

\section{Introduction}
\label{sec:intro}
Community detection is a central problem in network analysis and machine learning, with wide-ranging applications in social networks, image segmentation, biological networks, and beyond---see~\cite{abbe2018community} for a survey on the problem. The goal of community detection is to partition a network so that nodes in the same community are well-connected to each other and nodes across different communities are sparsely connected. A classic and elegant formalization of this problem is the exact recovery problem in the stochastic block model (SBM)~\citep{holland1983stochastic}. In its simplest form, an SBM partitions vertices into two communities, with edges forming within communities with probability $p$ and across communities with probability $q$, where $p > q$. Recovering communities in a graph generated from an SBM has been an active area of research, with the exact conditions for recoverability well understood in terms of  $p$ and $q$~\citep{hajek2016achieving,MR3520025-Guedon16,montanari2016semidefinite, mcsherry2001spectral, moitra2016robust,MR4115142,ding2022robust,Liu-Moitra-minimax,abbe2020entrywise,wang2020nearly,cohen2022community}. 

However, in many real-world applications, the  edges that define the network structure are private and sensitive information. This is especially true in domains like bioinformatics, network science, fraud detection, and social network analysis.  Protecting the privacy of these edges is an important and central challenge, as revealing connections between specific nodes could expose confidential relationships or sensitive attributes. For this reason, there has been growing interest in differentially private versions of community detection algorithms, particularly those focused on edge-differential privacy~\citep{nissim2007smooth}. 

There are two main ways to obtain a private algorithm for a graph problem. One can design a novel private (problem-specific) algorithm, or one can apply a privatization technique to make the graph itself private (input perturbation).  The advantage of the latter approach is that it is general and can solve multiple  problems at once. Once the graph is privatized, one is free to apply \emph{any} non-private algorithm to it without incurring further privacy cost. Then one can reuse the vast array of non-private algorithms designed in the past decades for many problems. While this approach is powerful---and for this reason extensively studied, e.g., ~\citep{blocki2012johnson, gupta2012iterative, eliavs2020differentially}---it has a major drawback: privatizing the graph typically requires significantly more noise than solving a specific problem privately. 

Consider the case of differentially private (DP) SBM algorithms. The best known private algorithms can recover a correct solution for a regime of the model close to the information theoretic threshold for even a constant privacy budget  $\eps \in O(1)$. These algorithms, however, are quite complex and have running time $> n^2$, e.g.,~\cite{chen2023private}, where $n$ is the number of nodes in the graph. On the other hand, the only known graph privatization method for the SBM is the naive edge-flipping method---applying the randomized response mechanism~\citep{Wang17} on every edge and non-edge. This means that for $\eps$-DP, one makes each edge a non-edge and vice versa with probability $1/(1+e^\eps)$,  resulting in $\Omega(n^2)$ spurious edges (and a prohibitively dense graph) when $\eps \in O(1)$. Unfortunately, 
it is easy to verify that this large flipping probability is necessary
for this simple edge-flipping mechanism to be $\eps$-DP.  This makes recovering the optimal solution impossible for a large regime of parameters of the SBM as shown in prior work~\citep{seif2022differentially} unless $\eps \in \Omega(\log n)$. Such  a high $\eps \in \Omega(\log n)$ is usually unacceptable from a privacy perspective, as in this regime a private edge can be leaked (as is) with $1-1/\poly(n)$ probability. From a practical point of view, real-world deployments of DP strive to have $\eps \le 1$, and $\delta < 1/n$ for $(\eps,\delta)$-DP~\citep{near2023guidelines}. For this reason, we are interested in input-perturbation methods allowing exact recovery close to the information theoretic threshold for $\eps \in O(1)$ and $\delta \in O(1 / \poly(n))$. 

In this paper, we pursue a different input-perturbation approach that retains the best aspects of graph privatization without the large costs. We show that, under some constraints on how the  algorithm accesses the graph, it is possible to use the simple edge-flipping method to output a \emph{(semi)-}private graph with \emph{significantly less noise}.
We then prove that the \emph{output} of  an algorithm from this restricted class when applied to the \emph{(semi)-}private graph is differentially private.

More formally, we consider a class of algorithms that access the graph solely by repeatedly (and possibly adaptively) computing the number of edges between a node and a sufficiently large subgraph. 
We show that if an algorithm from this class executes these degree operations on a slightly perturbed graph 
(without any other modifications) then its output is 
differentially private.
The amount of these perturbations depends on the desired privacy guarantee and the size of the subgraphs to which degrees are measured.
Crucially, this is achieved while flipping the edges and non-edges of the graph with a much smaller probability than the naive edge-perturbation method. More precisely, our method flips each edge and non-edge with probability $O(\log(1/\delta)/(\ell\eps^2))$ where $\ell$ is a lower bound on the size of the subgraphs on which the degrees can be privately computed---in our SBM algorithms, $\ell \in O(n/\sqrt{\log n}$). 
In the regime of interest where $\delta \in \Theta(1/\poly(n))$ and for the range of $\ell$ used in our paper, we use a flipping probability of magnitude $\tilde \Theta\left(\frac{1}{n \eps^2}\right)$ versus $\Theta(1/\exp(\eps))$ for the naive edge-flipping algorithm ($n$ times smaller). This allows us for $\eps \in 
 \Omega (1/\poly(\log n))$ to introduce only $\tilde{\Theta}(n)$ spurious edges versus $\tilde{\Theta}(n^2)$ spurious edges introduced by the naive edge-flipping method.

Our approach offers several advantages, including retaining the sparsity of the graph (which is crucial for efficient implementations on large graphs) and the simplicity of the final algorithm. As an application of this framework, we show how this restricted access to graph data is sufficient to correctly recover partitions in the SBM with almost asymptotically optimal parameters for constant $\eps$ and $\delta \in \Theta(1/\poly(n))$ for both the classic SBM and the directed SBM (DSBM), for which we give the first DP algorithm.
Our main theorem is as follows.

\begin{theorem}(Informal version of Theorems~\ref{thm:main-accuracy-dsbm} and~\ref{thm:sbm-main})\label{thm:main}

For $\eps,\delta > 0$, there exist  $(\eps,\delta)$-DP algorithms that when provided with a graph from the SBM (or DSBM) with $\frac{p-q}{\sqrt{p}} \in \Omega \left(\frac{\sqrt{\log n}}{\sqrt{n}}\right)$ and $p \in \Omega \left ( \frac{\log(1/\delta)\sqrt{\log n}}{\eps^2 n} \right)$, achieve exact recovery with constant probability.\footnote{Here $p$ and $q$ are the edge probabilities between same community and different community vertices, respectively, in the SBM and DSBM models.}
Moreover, the algorithms guarantee $(\eps,\delta)$-DP unconditionally for all graphs---including graphs not from the (D)SBM models---and require, for $\eps \in \Omega\left (\frac{1}{\poly(\log n)}\right), \delta \in  \Omega \left(\frac{1}{\poly(n)}\right)$, time and space $\widetilde{O}(|V|+|E|)$.
\end{theorem}

Using efficient privacy boosting techniques (see~\Cref{sec:boosting}) we show how to boost the probability of correctly recovering the exact partition to $1-1/\poly(n)$ with only a constant factor increase in $\eps$ and maintaining $\delta \in  \Omega (1/{\poly(n)})$. We also extend the results of Theorem~\ref{thm:main-accuracy-dsbm} to the more general degree-heterogeneous stochastic block model (DHSBM) in \Cref{sec:dhsbm}.

These results allow us to obtain near-linear time algorithms for exact recovery for the privacy setting of interest of $\eps\in \Theta(1)$ and $\delta \in \Theta(1/ \poly(n))$ in the regime where  $\frac{p-q}{\sqrt{p}} \in \Omega \left(\frac{\sqrt{\log n}}{\sqrt{n}}\right)$, asymptotically matching the information-theoretic threshold, and where the only (asymptotic) cost of privacy is requiring $p\in \Omega \left ( \frac{\log n\sqrt{\log n}}{n} \right)$---that is only $\sqrt{\log n}$  larger than the graph-connectivity threshold of  $\Omega \left(\frac{\log n}{n}\right)$ necessary even without privacy. No prior private algorithms achieved exact recovery in this regime with sub-quadratic running time. 

To demonstrate the usefulness and efficiency of our framework we report experimental results on directed SBM graphs in \Cref{app-experiment-sbm}. Moreover, to show the adaptability of our framework on other problems, we report results in ~\Cref{app-experiment-degree} for an  application to the problem of outputting a graph’s degree sequence with edge differential privacy.

\paragraph{Technical overview.}
We now introduce the main technical components of our work. In Section~\ref{sec:edgeflip}, we define our novel approach for designing private algorithms with edge flipping.
The key technical challenge is to show that there exists a class of algorithms that is large enough to be useful to solve important problems like exact recovery but at the same time is sufficiently constrained to allow input perturbation with a lower amount of noise. While an unconstrained algorithm would force us to apply the same (high) edge-flipping probability of randomized response, we show that algorithms executing only degree counting operations on large subgraphs are inherently more stable to an edge addition and thus require less noise for the same privacy protections. Interestingly, we observe how this (more) constrained class of algorithms is sufficient to perform (approximately) all the cut-based operations employed to solve exact recovery in the SBM by prior work such as~\cite{cohen2022community}. 
A technical challenge in proving privacy for such a class of algorithms is that they execute degree counting operations that are chosen by the algorithm adaptively (i.e., based on the result of prior queries). In Section~\ref{sec:edgeflip}, we show, however, that as long as these degree-based queries are performed on edge-disjoint subgraphs, we do not have any loss in privacy (i.e., arbitrary many adaptive queries can be performed on the same perturbed graph with the privacy cost of one query). This is formalized in our framework defining the concept of degree counting in disjoint-stars (see Definition~\ref{def:disjoint-star-ada}.) 

In Section~\ref{sect:DSBM}, we show that our framework directly leads to an exact recovery algorithm for a directed version of the SBM. The key to this analysis is to show that after applying our (low noise) edge-perturbation method to a directed SBM graph with parameters $p$ and $q$, we obtain a graph that is distributed again according to a directed SBM, but with slightly different parameters $p'$ and $q'$ which are close to $p$ and $q$. Then we can apply the algorithm of \cite{cohen2022community} on our perturbed graph, thereby making this algorithm private while degrading its recovery threshold only slightly (and keeping its simplicity!). For full generality, we extend our results to the degree-heterogenous SBM, where each edge is drawn with its own independent probability, in \Cref{sec:dhsbm}.

In Section~\ref{sect:SBM},
we show a similar result to the directed SBM for the classical (undirected) SBM. 
We essentially reduce the undirected case to the directed case via a careful coupling argument.
This creates some subtleties since: (1) to preserve privacy, we 
have to add independent noise every time we examine the same edge and (2)
we have to show that correlations do not hurt our ability to recover the communities.
Then, in Section~\ref{app-experiment-sbm} we perform an empirical analysis of our algorithms showing that they can scale exact recovery on graphs with two orders of magnitude more nodes than prior work.

Finally, in \Cref{sec:boosting}, we show how to boost
our algorithms so they output the correct partition with high probability with just a constant increase in the privacy budget. Notice that given an algorithm that finds the correct partition with constant probability, it is easy to boost the probability of success to $1-1/\poly(n)$ by 
executing $O(\log n)$ parallel repetitions of the algorithm. However, when the algorithm has to be private, these repetitions would result
in an undesirable $\Omega(\log n)$ factor increase in the privacy loss. Requiring this larger privacy budget would cause
the threshold on $p$ required to guarantee exact recovery to increase by an $\Omega(\log^2(n))$ factor, moving the algorithm further away from the information-theoretic threshold. 
We avoid this degradation by using recent results on boosting with privacy amplification from \cite{cohen2022generalized}.
This enables us to boost the success probability while incurring only a constant factor increase in the privacy budget and a small sublinear factor
increase in running time. Obtaining this boosting result requires carefully separating the sources of randomness affecting the output of our algorithm that can be sampled independently (i.e., the random choices of the algorithm) from the randomness of the SBM model (which is fixed upon the realization of the graph from the model and cannot be boosted).

\paragraph{Additional related work.}
Our main focus is on designing private algorithms~\citep{seif2022differentially,chen2023private} for the exact recovery problem in the SBM.

The (non-private) SBM recovery problem has received widespread attention in the past~\citep{hajek2016achieving,MR3520025-Guedon16,montanari2016semidefinite, mcsherry2001spectral, moitra2016robust,MR4115142,ding2022robust,Liu-Moitra-minimax,abbe2020entrywise,wang2020nearly,cohen2022community} for its crisp mathematical characterization of the community detection problem. We refer to~\citep{abbe2018community} for a survey of results. It is well-known that exact recovery of the communities requires a signal-to-noise ratio of $\frac{p-q}{\sqrt{p}} \in O(\sqrt{\log n}/\sqrt{n})$. More precisely, for the regime of two communities, and $p=a\frac{\log n}{n}$, $q=b\frac{\log n}{n}$, non-private exact recovery requires $\sqrt{a}-\sqrt{b} > \sqrt{2}$. The majority of the literature has been devoted to defining spectral (e.g.,~\cite{mcsherry2001spectral,wang2020nearly}) or semi-definite programming (SDP)-based algorithms (e.g.,~\cite{hajek2016achieving,montanari2016semidefinite}) that can solve the problem close to the threshold in polynomial time. Recent results in the non-private literature can achieve exact recovery in near-linear time~\citep{wang2020nearly,cohen2022community}. Our work in this paper is based on techniques from~\citep{cohen2022community}, in particular, showing that a variant of such an algorithm can be defined (and analyzed) as part of our novel graph perturbation framework.   

\paragraph{Private recovery in the stochastic block model.}
Despite its long history, the study of private algorithms for the SBM is quite recent. One of the first result in the area is the work of~\citep{seif2022differentially}. This work---and later follow-ups~\citep{seif2024differentially}---provided edge-differentially private (quasi-)polynomial time algorithms for the regime $\sqrt{a}-\sqrt{b} 
\in \Theta_\eps(1)$, very close to the recovery threshold. These non-poly-time algorithms are based on semi-definite programming and maximum likelihood estimation. \citep{seif2022differentially} also provided a poly-time randomized response-based algorithm that achieves exact recovery for the same regime provided that $\eps \in \Omega(\log n)$.

More recently~\citep{chen2023private} provided the first poly-time edge-differentially private algorithms that asymptotically reach the optimal recovery threshold for constant $\eps$. More precisely, consider the privacy setting of $\eps \in \Theta(1)$, $\delta \in \Theta(1/\poly(n))$. For the SBM regime with $p\in \Omega((\log n)/n)$, this SDP-based algorithm achieves exact recovery for $(p-q)/\sqrt p \in \Omega \left (\frac{\sqrt n}{\sqrt{\log n }}\right )$. 
Later~\citep{nguyen2024differentially}, again using a SDP algorithm, improved over the asymptotic constant of the SBM regime. Both of these algorithms require $\Omega(n^2)$ running time. 

Comparing our results (see Theorem~\ref{thm:main}) again for $\eps \in \Theta(1)$, $\delta \in \Theta(1/\poly(n))$, we achieve  (asymptotically) the same  threshold of \citep{chen2023private,nguyen2024differentially} for $(p-q)/\sqrt{p}$  while requiring only \emph{near-linear} time in the graph size. This is achieved at the sole cost of requiring a slightly denser (by still very sparse) input graph $p\in\Omega \left(\frac{{\log^{1.5} n}}{ n}\right)$ (in our algorithm) versus $p\in\Omega \left(\frac{{\log n}}{ n}\right)$ for the optimal threshold as in~\citep{chen2023private,nguyen2024differentially}. We observe that no prior work achieved exact recovery in near-linear running time for any regime of the SBM when $\eps \in O(1)$.

Other variants of the problem have also been studied. \cite{imola2023differentially} focused on a hierarchical version of the SBM, for which it provided edge-differentially-private spectral algorithms.   \cite{he2024differentially,chen2023private} provided private algorithms for the weak recovery problem (where one accepts some misclassified nodes). \cite{guo2023privacy} studied the problem in a federated (distributed) setting.
All results mentioned so far are for the edge-DP setting. In the more challenging node-DP setting \citep{graphon_CDDHLS24}, recently studied the problem of recovering the parameter of the SBM (as opposed to the actual partitions)---this is know as the graphon estimation problem. 

\section{Preliminaries}\label{sec:prelim}
In this paper, we consider graphs in the standard (undirected) \defn{stochastic block model (SBM)} (a.k.a. the planted partition model) with two equally sized communities.
Let $p,q \in [0,1]$. Let $V$ be a set of nodes that is \defn{partitioned} into two communities of equal size, i.e.,  $V = V_1 \cup V_2$ where $V_1 \cap V_2 = \emptyset$ and $|V_1| = |V_2| = n/2$.
We say that random graph $G =(V,E)$ is drawn from the stochastic block model $\mathbf{SBM(V_1,V_2,p,q)}$ denoted $\mathbf{G\sim SBM(V_1,V_2,p,q)}$ if
for any two nodes $u \in V_i, v \in V_j$, with $i=j$, $\Pr[(u,v)\in E] = p$, i.e., there is an (undirected) edge between any two nodes of the same community with probability $p$; and
for any two nodes $u \in V_i, v \in V_j$, with $i\neq j$, $\Pr[(u,v)\in E] = q$, i.e., there is an (undirected) edge between any two nodes of the different communities with probability $q$.
We refer to $(V_1,V_2)$ as the \defn{planted partition} of $G$ (or of the vertex set $V$, depending on context.)

We additionally consider graphs in the \defn{directed stochastic block model (DSBM)}~\cite{DSBM_CL0Z20}. 
This model is analogous to the standard undirected SBM, except that \emph{outgoing} edges to nodes of the same community and opposite community are drawn, independently, with probabilities $p$ and $q$, respectively. 

We write $\mathbf{Bernoulli(p)}$ to denote the \defn{Bernoulli distribution}, which takes value 1 with probability $p$ and 0 with probability $1-p$. 
We write $\mathbf{B(n,p)}$ to denote the \defn{Binomial distribution}, which counts the number of 1s in $n$ independent trials of $\bern{p}$ random variables.
We also use the notation $\mathbf{X \sim A}$ to denote an instance $X$ drawn from the distribution of random variable $A$. 

For graph $G=(V,E)$, subset $S \subseteq V$ and $u \in V$, the \defn{neighborhood of $u$ in $S$} is the set $\mathbf{N_G(u,S)} := \{(u,v) \in E \textrm{ } | \textrm{ } v \in S \}$, and we write $\mathbf{n_G(u,S)} = |N_G(u,S)|$ and refer to this as the \emph{degree of $u$ to $S$}. 
When working in the DSBM, $N_G(u,S)$ denotes the set of \emph{outgoing} edges from $u$ to $S$. 
We also use the notation $\mathbf{X \sim A}$ to denote an instance $X$ drawn from the distribution of random variable $A$.

We study the widely used notion of edge differential privacy for graphs. Algorithm $\alg$ is \defn{$\mathbf{(\eps,\delta)}$-edge differentially private} if for any two (directed or undirected) graphs $G_1,G_2$ that differ by one edge, and any subset $X$ of the range of outputs of $\alg$, 
$$\Pr[\alg(G_1) \in X] \leq e^\eps \Pr[\alg(G_2) \in X] + \delta. $$
Since edge differential privacy is the only notion of privacy considered this paper, we will henceforth drop the `edge' modifier and refer to the property simply as \defn{$\mathbf{(\eps,\delta)}$-differential privacy ($\mathbf{(\eps,\delta)}$-DP)}.

It will be convenient for our exposition to assume that $\delta<1/n$ and $\eps \leq 12 \log (2/\delta)$. 
We note that larger values of $\eps, \delta$ would not result in meaningfully private algorithms, as the difference in outcomes between two adjacent graphs is much too large. Moreover, by definition, any algorithm that satisfies $(\eps,\delta)$-DP also satisfies $(\eps,\delta')$-DP for any $\delta'>\delta$, and similarly for $(\eps',\delta)$-DP with $\eps' > \eps$. Therefore, there is no loss of generality with this assumption.

Finally, we will make use in this paper of the standard Chernoff bound.

\begin{lemma}[Chernoff bound, Theorems 4.4, 4.5 in \cite{book_prob_MU05}]\label{chernoff}
Let $X_1,\dots X_k$ be independent Poisson trials, and let $X = \sum_{i=1}^k X_i$ and $\mu = \E[X]$. Then for $0 < \delta \leq 1$,
\[ \Pr(X \geq (1+\delta)\mu) \leq e^{-\mu\delta^2/3}\]
and 
\[ \Pr(X \leq (1-\delta)\mu) \leq e^{-\mu\delta^2/3}.\]
\end{lemma}

\section{Private degree-based algorithms with low noise perturbations}\label{sec:edgeflip}

In this section, we introduce our input-perturbation-based framework for designing private algorithms. Our mechanism is based on applying the edge-flipping mechanism from randomized response on the input 
graph but with a significantly smaller flipping probability than what is required for publishing an edge-DP graph. While the flipping probability we use is insufficient 
to obtain a DP version of the graph as in the previous uses of this method~\citep{dp_sbm_SNVT22}, we show that this limited amount of noise is sufficient to prove DP for the output of any algorithm in a family of algorithms that access the graph in a restricted way (essentially making degree-based  operations).
We now introduce the edge-flipping mechanism.

\begin{definition}\label{def:edge-flipping-mechanism}
The \defn{edge-flipping mechanism} $\edgeflip_{\po}$ with parameter $\po \in [0,1]$ for directed graphs is defined as follows. Given directed graph $G=(V,E)$, $\edgeflip_{\po}(G)$ creates the random graph $H=(V,E')$ such that, independently for each ordered pair of nodes $u,v \in V$, $u \neq v$:
(i) if $(u,v) \notin E$, $(u,v) \in E'$ with probability $\po$; and
(ii) if $(u,v) \in E$,  $(u,v)\notin E'$ with probability $\pl$ (i.e., $(u,v) \in E'$ with probability $1-\pl$). 
\end{definition}

Our key insight is that for an arbitrary subset $S \subset V$ larger than a certain bound (depending on the desired privacy guarantees) and a node $u \in V \setminus S$, it is enough to have a small probability of perturbing each edge to compute the number of edges between $u$ and $S$ with DP guarantees. More precisely we can state the following lemma.

\begin{lemma}\label{lemma:dp-degree}
Let $G=(V,E)$ be a directed graph. Given $\eps,\delta, \ell > 0$, let
$\po = \min\left(\frac{96 \log(2/\delta)}{\ell\eps^2},1/2\right)$.
Let $H = \edgeflip_{\po}(G)$ be the output of the edge-flipping mechanism in Definition~\ref{def:edge-flipping-mechanism}.  
Then given any $S \subset V$ with $|S| \geq \ell$ and any $u \in V \setminus S$, the mechanism that outputs the neighborhood size $n_H(u,S)$ 
is $(\eps,\delta)$-DP.
\end{lemma}

\begin{proof}
Let $G = (V,E)$ and $G' = (V,E')$ be graphs such that $E$ and $E'$ differ by one edge, and fix $u \in V$.
Write $H = \edgeflip_{\po}(G)$ and $H' = \edgeflip_{\po}(G')$.

Observe that $N_G(u,S)$ and $N_{G'}(u,S)$ can differ by at most one element, and the case that $N_G(u,S) = N_{G'}(u,S)$ is immediate. Therefore, we need only consider the case $n_G(u,S) = x = n_{G'}(u,S) - 1$ for some integer $x$ (the proof for the case in which $n_{G'}(u,S) = y = n_{G}(u,S) - 1$ is identical).

Since $M$ performs a sequence of Bernoulli trials on the edges and non-edges in a graph, we can write 
$ n_H(u,S) = x - B_{minus} + B^0_{plus}$,
where $B_{minus}$ and $B^0_{plus}$ are independent random variables such that $B_{minus} \sim \bin{x}{\pl}$ 
and $B^0_{plus} \sim \bin{|S| - x}{\po}$.

We rewrite this as 
$n_H(u,S) = x - B_{minus} + B_0 + B_{plus}$. 
Where $B_0 \sim \bern{\po}$, $B_{plus} \sim \bin{|S|-x-1}{\po}$,  and $B_{minus}, B_0,$ and $B_{plus}$ are independent.

Similarly we write 
$n_{H'}(u,S) = x + 1 -  B'_{minus}  - B'_0 +B'_{plus}$.
Where $B'_{minus} \sim \bin{x}{\pl}$, $B'_0 \sim \bern{\pl}$, and $B'_{plus} \sim \bin{|S|-x-1}{\po}$ are independent random variables.

It remains to show that 
$n_H(u,S)$ and $n_{H'}(u,S)$ are $(\eps,\delta)$-close. 
We prove the following claim in \refappendix{Appendix~\ref{app:proof-longcalcs}}.

\begin{claim}\label{lemma:dp-long-calcs}
For all sets $F\subseteq \mathbb{N}$,
$\Pr[n_H(u,S) \in F] \leq e^{\eps}  \Pr[n_{H'}(u,S) \in F] + \delta$ and 
$\Pr[n_{H'}(u,S) \in F] \leq e^{\eps}  \Pr[n_{H}(u,S) \in F] + \delta.$
\end{claim}

With the claim, it follows that the mechanism that outputs $n_H(u,S)$ is $(\eps,\delta)$-differentially private.
\end{proof}

\subsection{Framework for private graph algorithms}
Lemma~\ref{lemma:dp-degree} shows that it is possible to estimate
the degree  $n_H(u,S)$ of a node $u$ to a subset $S$ in an edge-flipped directed graph $H$
 with DP guarantees given appropriate flipping probability $\po$. In graph algorithms, however, we are often interested in computing many such degrees for various nodes and sets. In this section, we show that a simple application of parallel composition for DP allows us to obtain multiple node degrees in various (disjoint) subsets without additional privacy loss. 

\begin{definition}[Degree counting in disjoint stars]\label{def:deg-counting-star}
Let $G=(V,E)$ be a directed graph. A set $\mathcal{S} = \{(u_i, S_i) \mid  i \in [t]\}$, for $t \ge 1$, is called a \emph{disjoint-star instance} if $\forall i \in [t]$, $u_i \in V, S_i \subset{V}, u_i\notin S_i$ and for all $j \neq i \in [t]$, 
$$\{ \{ u_i \} \times S_i\} \cap \{ \{u_j\} \times S_j\} = \emptyset,$$
where $A \times B$ is defined as the set of unordered pairs $(a,b)$ with $a\in A$ and $b \in B$.
The problem of degree counting in disjoint stars is defined as computing $n_G(u_i,S_i)$ for all $i\in [t]$.
\end{definition}

\paragraph{Adaptive disjoint-star algorithm.}

We are now ready to prove the main result of the section: graph algorithms that operate only by (sequential) accesses to degrees of disjoint stars  in $\edgeflip_{\po}(G)$ are DP.

\begin{definition}[Adaptive disjoint-star algorithm\label{def:disjoint-star-ada}] Given a (directed) graph $G=(V,E)$, and $\eps,\delta,\ell > 0$, let $\po =  \min\left(\frac{96 \log(2/\delta)}{\ell\eps^2},1/2\right)$ and $H = \edgeflip_{\po}(G)$.
An $\ell$-bounded adaptive disjoint-star algorithm is defined by a sequence of $t$ randomized algorithms $A_i$ for $i\in [t]$, each of which is used to obtain a star $(u_i, S_i)$ in the sequence, for $u_i \in V$ and $S_i \subseteq V$. 
Let $O_i = n_H(u_i, S_i)$. Each algorithm $A_i$ takes as inputs $V$, all prior stars $S_1, \ldots, S_{i-1}$ and the prior computed star degrees $O_1, \ldots, O_{i-1}$ in $H$ and (using its own randomness $R_i$) outputs a star $(u_i, S_i)$ with $|S_i| \ge \ell$ disjoint from the prior stars.  
\end{definition}

The following theorem follows by noting that each (or non-edge) in $G$ participates in a single degree counting computation, and applying Lemma~\ref{lemma:dp-degree} on disjoint stars. 

\begin{theorem}(Privacy of adaptive disjoint-star algorithms)\label{thm:privacy-adptive-stars}  
    Given a directed graph $G$ and $\delta, \eps, \ell > 0$,
    any $\ell$-bounded disjoint-star algorithm run on $G$ which outputs $S_1, \ldots, S_t, O_1, \ldots, O_t$ is $(\eps,\delta)$-DP.
\end{theorem}

\begin{proof}
    Let $G$, and $G'$ be adjacent graphs differing by a single edge $(u,v)$. Without loss of generality let $G=(V,E)$ and $G'=(V,E \; \cup \; \{(u,v)\})$. So $(u,v)$ is a non-edge in $G$ and an edge in $G'$.
Let $\mathcal{S} = (u_1, S_1), \ldots, (u_t, S_t)$ be an arbitrary disjoint star sequence and let $\mathcal{O} = O_1,\ldots, O_t$ be an arbitrary degree output sequence. 

Notice that by disjointedness, $(u,v)$ can participate in at most one star in $\mathcal{S}$ for any execution of the algorithm. Let $E_{\mathcal{S},\mathcal{O}}$ be the event that the (adaptive) algorithm executes a sequence of disjoint star degree queries $\mathcal{S}$ and obtains output $\mathcal{O}$.  We say that $(u,v)\in \mathcal{S}$ to indicate that there is a star $(a, B)$ in $\mathcal{S}$ such that $(u,v) \in a \times B$. 
We say $(u,v)\notin S$ otherwise.

If $(u,v)\notin \mathcal{S}$ then $$\Pr_G(E_{\mathcal{S},\mathcal{O}}) = \Pr_{G'}(E_{\mathcal{S},\mathcal{O}}).$$ 
This holds since the algorithm is executed on the same data (which does not contain the edge or non-edge $(u,v)$) in both $G$ and $G'$.

We now focus on a  subset $X$ of possible outputs of the algorithm such that in each output in $X$ the edge $(u,v)$ belong to exactly one star. Let $E_X$ be the event that output of the algorithm is in $X$. 
Let $Q_{S_i}$ be the event that the star to which $(u,v)$ belongs is $S_i = (a_i, B_i)$.  We have  that
\begin{equation}
\label{eq:QS}
\Pr_G(Q_{S_i}) = \Pr_{G'}(Q_{S_i})
\end{equation}
for any $S_i$. 
This is because, up to the point where the first (and only) star involving $(u,v)$  is selected, the subgraphs of $G$ and $G'$ that the algorithm observes are identical.
Let $A_{O,S_i}$  be the event that the degree output for the (only) star $S_i$ involving  $(u,v)$ is $O$. 

Also, notice that conditioned on $Q_{S_i}, A_{O,S_i}$ the remaining execution of the algorithm on $G$ and $G'$ is identical, hence 
\begin{equation}
\label{eq:AO}
    \Pr_{G} (E_{X} \mid Q_{S_i}, A_{O,S_i}) = \Pr_{G'} (E_{X} \mid Q_{S_i}, A_{O,S_i}).
\end{equation}

 Using Equations (\ref{eq:QS}) and (\ref{eq:AO}) we get
 \begin{equation*}
 \begin{split}
Pr_G(E_{X}) = \sum_{S_i} \sum_O \Pr_{G} (E_{X} \mid Q_{S_i}, A_{O,{S_i}})\Pr_{G}(A_{O,{S_i}} \mid Q_{S_i}) \Pr_{G}(Q_{S_i}) \\
= \sum_{S_i} \sum_O \Pr_{G'} (E_{X} \mid Q_{S_i}, A_{O,{S_i}})\Pr_{G}(A_{O,{S_i}}\mid Q_{S_i}) \Pr_{G'}(Q_{S_i}).
\end{split}
\end{equation*}

By Lemma~\ref{lemma:dp-degree} and by the observation that the output degree of a star $S_i$ depends only on ${S_i}$ it follows that for any $S_i,O$
$$e^{-\eps}  \Pr_{G'}(A_{O,{S_i}} \mid Q_{S_i}) - \delta \le  \Pr_{G}(A_{O,{S_i}}|Q_{S_i}) \le e^{\eps}  \Pr_{G'}(A_{O,{S_i}}\mid Q_{S_i}) + \delta.$$

Hence,
$$\Pr_G(E_{X}) \le  
  \sum_{S_i} \Pr_{G'}(Q_{S_i}) \sum_O \left(e^{\eps} \Pr_{G'}(A_{O,{S_i}}\mid Q_{S_i}) + 
  \delta\right) \Pr_{G'} (E_{X} \mid Q_{S_i}, A_{O,{S_i}})  $$
$$\Pr_G(E_{X}) \le  
  e^\eps \sum_{S_i} \Pr_{G'}(Q_{S_i}) \sum_O  \Pr_{G'}(A_{O,{S_i}}\mid Q_{S_i}) \Pr_{G'} (E_{X} \mid Q_{S_i}, A_{O,{S_i}})  +\delta \sum_{S_i} \sum_O  \Pr_{G'}(Q_{S_i}) \Pr_{G'}(A_{O,{S_i}}\mid Q_{S_i})$$
$$\le  e^{\eps}\Pr_{G'}(E_{X}) + \delta$$
The other direction follows by the same argument.
\end{proof}

\paragraph{Efficient implementation.}

We now show that it is possible to implement the mechanism of Definition \ref{def:edge-flipping-mechanism} efficiently. In~\refappendix{Algorithm~\ref{alg:flip}} 
we give the pseudocode of
 an edge-flipping procedure that 
 runs in near-linear time (in the size of its input and output). This procedure is similar to the one in~\cite{cormode2012differentially}.\footnote{The ability to generate an edge-flipped graph efficiently allows execution of a (disjoint star based-) algorithm on such graphs without modifying the algorithm (i.e., treating the graph algorithm as a black-box). However, it is easy to see that one does not need to generate the edge-flipped graph to obtain the degrees of the stars if it is possible to modify the graph algorithm to add noise to the degree counts. Indeed, it is also possible to simulate the process in Lemma~\ref{lemma:dp-degree} by adding appropriate binomial noise to the true degrees.}

We make the standard assumption that sampling a Bernoulli or a Binomial variable can be done in $\widetilde{O}(1)$ time, where $\widetilde{O}(1)$ hides factors polylogarithmic in $|V|$.

\begin{lemma}\label{lemma:edge-flipping-time} 
With high probability in $n$, the edge-flipping mechanism $\edgeflip_{\po}$ can be implemented to run in $ \widetilde{O}(|V|+|E|+|V|^2\po)$ time and requires 
$ \widetilde{O}(|V|+|E|+|V|^2\po)$ space.
\end{lemma}

\begin{proof}
Notice that deciding which of the edges in $E$ are preserved takes $\widetilde{O}(|E|)$ time. We now consider two cases. In the dense case, when $|E| > \frac{1}{100}|V|^2$,
we implement the flipping of the non-edge by
 simulating directly the independent Bernoulli trials on all non-existing edges. This takes $\widetilde{O}(V^2)$ which is  $\widetilde{O}(|E|)$. 

Consider the case $|E| \leq \frac{1}{100}|V|^2$. If $\po > \frac{1}{100}$ the output is dense so we apply the independent Bernoulli trials with $\widetilde{O}(V^2)$ time. 

Finally, consider the case $|E| \leq \frac{1}{100}|V|^2$ and $\po \leq \frac{1}{100}$.
A simple application of a Chernoff bound shows that with high probability $(m_{+}) + |E| \leq \frac{4}{100}|V|^2.$ Hence we can 
efficiently sample edges from $V \times V$ until we find $m_{+}$ new edges (not part of $E$ or previously sampled) via rejection sampling. This takes $O(|V|^2\po)$ time with high probability (see~\cite{cormode2012differentially}).  
\end{proof}

As we show later, for the settings of interest of the SBM problem and privacy parameters, the flipping probability $\po$ satisfies $\po \in \tilde O (1/|V|)$, resulting in $ \widetilde{O}(|V|+|E|)$ space and running time. For pseudocode of the efficient implementation of $\edgeflip_{\po}$,  refer to~\refappendix{Appendix~\ref{app:algo-flip}}.

\section{Directed Stochastic Block Model}\label{sect:DSBM}
We now give our algorithm for differentially private community detection in the directed SBM (DSBM). We will later use this algorithm to solve private community detection in the undirected SBM in  Section~\ref{sect:SBM}.

Our algorithm for the DSBM is based on the algorithm from \cite{colt22directed}, which studies the more general degree-heterogeneous stochastic block model (DHSBM) on directed graphs. In this model, each node $u$ has its own probabilities $p_u,q_u$ of outgoing edges to nodes from its own community and the opposite community, respectively. 
Their results naturally generalize to our setting by setting all $p_u$ to the same value $p$ and all $q_u$ to the same value $q$. 
We present our main results for the homogeneous directed and undirected SBMs, as these are the most widely used SBM models, as well as to simplify the presentation.
For full generality, we extend our results to the heterogeneous SBM setting in \Cref{sec:dhsbm}.

We give the pseudocode for our private exact recovery algorithm in Algorithm~\ref{alg:dir_sbm}. This algorithm uses the subroutines $\mathrm{Update}$ and $\mathrm{Phase 1}$, provided in~\refappendix{Appendix~\ref{app:subroutines}}.
At a high level, after applying our private edge-flipping mechanism, the algorithm first splits the nodes of the graph into two evenly-sized subsets $S,S'$, and calls $\mathrm{Phase 1}$ on them separately. $\mathrm{Phase 1}$ randomly splits nodes of its input into $\sqrt{\log n}$ subsets $S_j$. Then, it iteratively classifies the nodes of a set $S_j$ into a partition based on counting degrees to the two subsets in a bipartition of another set $S_i$ (one such round of classification on a pair $S_i,S_j$ is done by the $\mathrm{Update}$ subroutine). 
The (initially random) partitions of the $S_i$ get progressively refined into more pure partitions (i.e., each partition containing more and more nodes from one community only). \footnote{Specifically, we stipulate in the pseudocode that the order of the subsets $S_1,\dots,S_b$ should follow an Eulerian tour on the complete graph on $b$ vertices. This ensures that after using $S_a$ to classify $S_b$, the next step is to use $S_b$ to classify some $S_c$ such that the pair $S_b,S_c$ has never been examined. (As such a tour only exists for complete graphs on odd numbers of vertices, we also replace $b$ with $b+1$ for even $b$.) In practice, any efficient algorithm for finding such a tour can be used, such as Hierholzer’s Algorithm, which runs in $O(b^2) = O(\log n)$ time (see e.g., \cite{hierholzer}).}
In the second phase, we run $\mathrm{Update}$ again to classify the nodes of subset $S$ based on degree counts to the final $\mathrm{Phase 1}$ partitions of $S'$, and vice versa. By the end, the partitions are exactly aligned with the ground truth communities with good probability.

Note that apart from a few additions, the rest of our algorithm is identical to the algorithm from \cite{colt22directed}. The first is, critically, the application of our private edge-flipping mechanism $\edgeflip_{\po}(G)$. Second, to guarantee privacy in our adaptive disjoint-star framework, we must ensure that the bipartitioned subsets of a set $S_i$ that are used to classify another set $S_j$ are \emph{deterministically} over the required size threshold. Thus, we add special handling to $\mathrm{Update}$ to ensure this minimum subset size is met, which adds some logical subtleties to the proofs. We give a brief overview here and describe the algorithms and proof steps in detail in \Cref{sec:dp-private}.

We give the analysis of our algorithm in two parts. First, in \Cref{sec:dp-int} we present an `intermediate' algorithm, Algorithm~\ref{alg:dir_int}, which applies edge flipping but is not yet fully private. It is convenient for exposition to separate out this algorithm, as we can isolate the effect of edge flipping and thus first prove utility on an algorithm which is closer to the algorithm from \cite{colt22directed}; this is done in \Cref{sec:dp-int}. The key difference in our fully private algorithm, Algorithm~\ref{alg:dir_sbm}, is that the minimum subset size is enforced at every step of vertex classification. This is done by moving vertices from one label set to the other if the minimum size is violated, which could destroy the utility guarantees, as vertices are no longer in the side of their partition corresponding to the larger neighborhood. However, in \Cref{sec:dp-private} we prove that with high probability, this minimum is not violated in any stage of the algorithm. Thus, we can conclude that Algorithm~\ref{alg:dir_sbm} has the same utility properties as Algorithm~\ref{alg:dir_int} with high probability. In \Cref{sec:dp-private}, we will prove the following theorems.

\begin{algorithm}[ht]
\noindent \textbf{Input:} A directed graph $G = (V, E)$ with $n$ vertices, privacy parameters $\eps,\delta$

\noindent \textbf{Output:} A bipartition of $V$

\begin{algorithmic}[1]
\STATE $\ell \gets \frac{n}{18\sqrt{\log n}}$
\STATE $\po \gets \min\left(\frac{96\log(2/\delta)}{\eps^2 \ell},1/2\right)$
\STATE Graph $H \gets \edgeflip_{\po}(G)$

\STATE $S,S' \gets$ a partition of $V$ randomly into two disjoint subsets of size $n/2$
\STATE $H_S \gets$ the subgraph of $H$ induced by $S$; $H_{S'} \gets$ the subgraph of $H$ induced by $S'$

\STATE $S_1,S_2 \gets \mathrm{Phase 1}(H_S,\ell)$
\STATE $S'_1,S'_2 \gets \mathrm{Phase 1}(H_{S'},\ell)$
\STATE $S^*_1,S^*_2 \gets \mathrm{Update}(H,\ell,S,S'_1,S'_2)$
\STATE $S^{'*}_1,S^{'*}_2 \gets \mathrm{Update}(H,\ell,S',S_1,S_2)$
\RETURN $(S^*_1 \cup S^{'*}_1, S^*_2 \cup S^{'*}_2)$

\caption{$\mathrm{PrivateDSBMRecovery}(G,\eps,\delta)$ an algorithm for exact recovery in the DSBM} \label{alg:dir_sbm}
\end{algorithmic}
\end{algorithm}

\begin{theorem}
\label{thm:dsbm-privacy}
On any graph $G$, Algorithm~\ref{alg:dir_sbm} is $(\eps,\delta)$-differentially private. 
\end{theorem}

\begin{theorem}\label{thm:main-accuracy-dsbm}
Let $V$ be a vertex set of size $n$ with planted partition $(V_1,V_2)$, and let $\delta, \eps > 0$.
Let $p,q>0$ such that $\frac{p-q}{\sqrt{p}} \in \Omega\left(\frac{\sqrt{\log n}}{\sqrt{n}}\right)$ and $p \in \Omega\left( \frac{\log(1/\delta)\sqrt{\log n}}{\eps^2 n}\right)$,
and suppose $n \in \Omega\left( \frac{\log(1/\delta)\sqrt{\log n}}{\eps^2 } \right)$.
Then 
with probability $1-1/poly(n)$ over the draw  from $DSBM(V_1,V_2,p,q)$ and the edge flipping  by $\edgeflip_{\po}$, we get a graph $H$
such that
Algorithm~\ref{alg:dir_sbm} recovers the underlying communities from $H$
with probability at least $1/3$ over the algorithm's partitioning randomness (and independently from the randomness 
of $DSBM(V_1,V_2,p,q)$ and the edge-flipping randomness).
Moreover, for $\eps \in \Omega(1/\poly(\log(n))), \delta \in  \Omega(1/\poly(n))$  the algorithm uses time and space $\widetilde{O}(|V|+|E|)$.
\end{theorem}

\subsection{Intermediate DSBM Algorithm}\label{sec:dp-int}
As previously discussed, in order to prove Theorem~\ref{thm:main-accuracy-dsbm}, we must first prove an equivalent utility theorem on an intermediate algorithm. Algorithm~\ref{alg:dir_int} is almost identical to Algorithm~\ref{alg:dir_sbm}, without the key addition of enforcing a minimum set size for each classification. Specifically, the function $\mathrm{Update}(G,\ell,B,A_1,A_2)$ classifies the nodes in set $B$ using the input partition $A_1$ of $A_2$ of set $A$. This is done by examining whether nodes in $B$ have more edges to $A_1$ of $A_2$, and then classifying them into $B_1$ or $B_2$, respectively. Therefore $A_1,A_2$ must be the same size, or else nodes could be biased to the larger set regardless of how the communities are distributed. Thus $\mathrm{Update}$ randomly subsamples the larger set to make it the same size as the smaller one. This is also what is done in \cite{colt22directed}, but since that paper is not concerned with privacy, it is not an issue there that this procedure does not guarantee the smaller subset to be of any minimum size. 

In our $\mathrm{Update}$ function, we add lines 2 -- 9 of the pseudocode to enforce exactly this by moving nodes from the larger set to the smaller one if the smaller set has size less than some parameter $\ell$. In practice, we set $\ell= \frac{n}{18\sqrt{\log n}}$ to apply the privacy lemma from our adaptive disjoint-star framework. Although we cannot guarantee utility once this procedure is invoked, we will later show that with high probability, this never happens over the course of the algorithm (nevertheless, whp would not have been sufficient to prove that our algorithm is DP on every input, hence this modification).

The full pseudocode for Algorithm~\ref{alg:dir_int} is presented in \Cref{app:intermediate}. The intermediate algorithm, Algorithm~\ref{alg:dir_int}, simply omits lines 2 -- 9 and is thus the same as the analogous function from \cite{colt22directed}. This allows us to appeal to the utility lemmas from that work, and first prove simply that edge flipping preserves sufficient utility. In \Cref{sec:dp-private}, we build on this to prove utility for the fully private algorithm.

\begin{lemma}\label{lem:int-accuracy}
Let $V$ be a vertex set of size $n$ with planted partition $(V_1,V_2)$, and let $\delta, \eps > 0$.
Let $p,q>0$ such that $\frac{p-q}{\sqrt{p}} \in \Omega\left(\frac{\sqrt{\log n}}{\sqrt{n}}\right)$ and $p \in \Omega\left( \frac{\log(1/\delta)\sqrt{\log n}}{\eps^2 n}\right)$,
and suppose $n \in \Omega\left( \frac{\log(1/\delta)\sqrt{\log n}}{\eps^2 } \right)$.
Then 
with probability $1-1/poly(n)$ over the draw  from $DSBM(V_1,V_2,p,q)$ and the edge flipping  by $\edgeflip_{\po}$, we get a graph $H$
such that
Algorithm~\ref{alg:dir_int} recovers the underlying communities from $H$
with probability at least $1/3$ over the algorithm's partitioning randomness (and independently from the randomness 
of $DSBM(V_1,V_2,p,q)$ and the edge-flipping randomness).
Moreover, for $\eps \in \Omega\left (\frac{1}{\poly(\log(n))}\right), \delta \in  \Omega \left(\frac{1}{\poly(n)}\right)$  the algorithm uses time and space $\widetilde{O}(|V|+|E|)$.
\end{lemma}

For correctness, we appeal to the result from \cite{colt22directed}, which requires first showing that our edge-flipping mechanism $\edgeflip_{\po}(G)$ produces a graph in the DSBM that retains the parameter requirements of $p,q$. 

\begin{lemma}\label{lem:noised-dsbm}
Let $V$ be a vertex set with planted partition $(V_1,V_2)$, and let $p,q,\po \in [0,1]$. 
For any graph $G=(V,E) \sim DSBM(V_1,V_2,p,q)$, the edge-flipped graph $H =\edgeflip_{\po}(G)$ is distributed according to $DSBM(V_1,V_2,p',q')$ for $p'=p(1-2\po) + \po$ and $q'=q(1-2\po) + \po$.
\end{lemma}

\begin{proof}
Let $G(V,E) \sim DSBM(V_1,V_2,p,q)$ and let $H=\edgeflip_{\po}(G)$. We must show that $H$ is distributed according to $DSBM(V_1,V_2,p',q')$.

Let $u,v \in V$ and suppose these two nodes are from the same community. Then $(u,v) \in E'$ if an only if either: (1) $(u,v)\in E$ and is not flipped, or (2) $(u,v)\notin E$, and this non-edge is flipped.
These events are mutually exclusive, and edges are flipped independently from the distribution of edges in $G$. 
Thus since $(u,v)$ is an edge in $G$ with probability $p$, and edges are flipped with probability $\po$, $(u,v)$ is an edge in $H$ with probability $p' = p(1-\po) + (1-p)\po = p(1-2\po) + \po.$

Now suppose $u,v$ are from different communities. 
Then similarly, $(u,v) \in E'$ if and only if either $(u,v)\in E$ is flipped, or $(u,v)\notin E$ and is flipped. An identical argument shows that this occurs with probability $q' = q(1-2\po) + \po$. 

Therefore, $H \sim DSBM(V_1,V_2,p',q')$, as desired.
\end{proof}

We now conclude the proof of \Cref{lem:int-accuracy}.

\begin{proof}
Theorem 1 from \cite{colt22directed} shows that the (non-private) exact recovery algorithm runs in time linear in $n$ and if $\frac{p-q}{\sqrt{p}} \in \Omega\left(\frac{\sqrt{\log n}}{\sqrt{n}}\right)$,  
it recovers the underlying communities of an input graph from $DSBM(V_1,V_2,p,q)$ with probability $1 - 1/\poly(n)$ on the randomness of $DSBM(V_1,V_2,p,q)$ and at least $1/3$ on the algorithm’s (partitioning) randomness. 
Our algorithm is identical to theirs after performing edge flipping, 
so it only remains to show for correctness that the edge-flipped graph $H =\edgeflip_{\po}(G)$ is distributed according to $DSBM(V_1,V_2,p'q')$ for some $p',q'$ satisfying $\frac{p'-q'}{\sqrt{p'}} \in \Omega\left(\frac{\sqrt{\log n}}{\sqrt{n}}\right)$.

Let $p',q'$ be as in Lemma~\ref{lem:noised-dsbm}. Algorithm~\ref{alg:dir_sbm} sets $\ell \gets \frac{n}{18\sqrt{\log n}}$ and $\po \gets \min\left(\frac{96\log(2/\delta)}{\eps^2 \ell},1/2\right) = \frac{1728 \log(2/\delta)\sqrt{\log n}}{\eps^2 n}$.
By selecting $n$ appropriately large, we can ensure that 
$\frac{1728 \log(2/\delta)\sqrt{\log n}}{\eps^2 n}  \leq c$ where $c < 1/2$ is a constant bounded away from $1/2$.
Hence $\po = \frac{1728 \log(2/\delta)\sqrt{\log n}}{\eps^2 n}$ and $1-2\po \geq 1-2c \in \Omega(1)$. This also implies that $p \in \Omega(\po)$ by the assumption on $p$, and so $p(1-2\po) + \po \in O(p)$. 
Putting these together, we have 
$ \frac{p'-q'}{\sqrt{p'}} = \frac{p(1-2\po) + \po - q(1-2\po) - \po}{\sqrt{p(1-2\po) + \po}} \in \Omega\left(\frac{(p-q)(1-2\po)}{\sqrt{p}}\right) = \Omega\left(\frac{p-q}{\sqrt{p}}\right)= \Omega\left(\frac{\sqrt{\log n}}{\sqrt{n}}\right).
$
The result then follows from Lemma~\ref{lem:noised-dsbm}.

Finally, assume $\eps \in \Omega\left (\frac{1}{\poly(\log(n))}\right), \delta \in  \Omega \left(\frac{1}{\poly(n)}\right)$. Then 
$\po = \frac{1728 \log(2/\delta)\sqrt{\log n}}{\eps^2 n} \in \tilde O \left(1/n)  \right) $ and hence by Lemma~\ref{lemma:edge-flipping-time}, the algorithm uses time and space $\widetilde{O}(|V|+|E|)$.
\end{proof}
\subsection{Private DSBM Algorithm}\label{sec:dp-private}
We can now prove privacy and utility for the fully private Algorithm ~\ref{alg:dir_sbm}.
In particular, we show that we can apply the utility  Lemma~\ref{lem:int-accuracy} from our intermediate Algorithm~\ref{alg:dir_int}, which is not fully DP. 

Differential privacy follows from our adaptive disjoint-star framework and the changes to the Update function that guarantee the minimum subset size required for privacy. 

\begingroup
\def\thetheorem{\ref{thm:dsbm-privacy}}
\begin{theorem}
On any graph $G$, Algorithm~\ref{alg:dir_sbm} is $(\eps,\delta)$-differentially private. 
\end{theorem}
\endgroup

\begin{proof}
The algorithm guarantees that subsets which are examined in any classification step are at least $\frac{n}{18\sqrt{\log n}}$ in size through lines 2 -- 9 of the pseudocode for the Update function. Thus Algorithm~\ref{alg:dir_sbm} satisfies Definition~\ref{def:disjoint-star-ada} of an adaptive disjoint star algorithm with $\ell = \frac{n}{18\sqrt{\log n}}$, and so privacy follows from Theorem~\ref{thm:privacy-adptive-stars}. 
\end{proof}

To prove utility, we will show that the procedure used to guarantee the minimum subset size of $\ell = \frac{n}{18\sqrt{\log n}}$ is never invoked with high probability. Then, with high probability, Algorithm~\ref{alg:dir_sbm} will exhibit the same utility as Algorithm~\ref{alg:dir_int} (as these two algorithms are identical save for this contingency procedure). Specifically, we show:

\begin{lemma}\label{lem:update-minsize}
With probability $1-1/n^{100}$ over the draw  from $DSBM(V_1,V_2,p,q)$ and the edge flipping by $\edgeflip_{\po}$, we get a graph $H$
such that over the course of the execution of 
Algorithm~\ref{alg:dir_sbm}, no labeled subset $S_{i,j}$ returned by a call to $\mathrm{Update}$ is of size less than $\ell = \frac{n}{18\sqrt{\log n}}$.
\end{lemma}

To prove Lemma~\ref{lem:update-minsize}, we first show the following lemmas. 

\begin{lemma}\label{lem:preference}
Let $G=(V,E) \sim DSBM(V_1,V_2,p,q)$  for $p>q$ be a graph, and let 
$A,B \subset V$ such that $A \cap B = \emptyset$.
Let $A_1, A_2 \subset A$ be a bipartition of $A$ with $|A_1| = |A_2| > 0$.
Then one of the following holds:
\begin{itemize}
\item for all $x \in B \cap V_1$, $P[(n_G(x,A_1) \geq n_G(x,A_2)] \geq 1/2$ and for all $y \in B \cap V_2$, $P[(n_G(y,A_2) \geq n_G(y,A_1)] \geq 1/2$; or
\item for all $x \in B \cap V_1$, $P[(n_G(x,A_2) \geq n_G(x,A_1)] \geq 1/2$ and for all $y \in B \cap V_2$, $P[(n_G(y,A_1) \geq n_G(y,A_2)] \geq 1/2$.
\end{itemize}
\end{lemma}

\begin{proof}
Suppose first that $|A_1 \cap V_1| = |A_2 \cap V_1|$, i.e., community $V_1$ is equally represented in both subsets $A_1$ and $A_2$. Since $A_1,A_2$ bipartition $A$, and since $V_1 \cap V_2 = \emptyset$, we also have that $|A_1 \cap V_2| = |A_2 \cap V_2|$. Let $x \in V_1, y \in V_2$. 
Since $H$ is drawn from $DSBM(V_1,V_2,p,q)$, and all intra-community edges are drawn independently and with the same probability, and all inter-communities edges are drawn independently and with the same probability, $P[(n_G(x,A_1) \geq n_G(x,A_2)] = P[(n_G(x,A_2) \geq n_G(x,A_1)] \geq 1/2$ and $P[(n_G(y,A_2) \geq n_G(y,A_1)] = P[(n_G(y,A_1) \geq n_G(x,A_2)] \geq 1/2$.
Now suppose that $|A_1 \cap V_1| > |A_2 \cap V_1|$, i.e., $A_1$ contains more nodes from community $V_1$ than from $V_2$. Then since $A_1,A_2$ have the same size and $V_1, V_2$ partition $V$, we must have $|A_2 \cap V_2| > |A_1 \cap V_2|$. Let $x \in V_1$. By the same properties of the DSBM as above, and the fact that $p > q$, $P[(n_G(x,A_1) \ge n_G(x,A_2)] \geq 1/2$.

Applying again the fact that $|A_2 \cap V_2| \ge |A_1 \cap V_2|$ and a symmetric argument, we also have that for all $y \in V_2$, $P[(n_G(x,A_2) \ge n_G(x,A_1)] \geq 1/2$. 

The other case of $|A_1 \cap V_1| < |A_2 \cap V_1|$ follows a symmetric argument to conclude that$P[(n_G(x,A_2) \ge n_G(x,A_1)] \geq 1/2$ for all  $x \in B \cap V_1$ and $P[(n_G(x,A_1) \ge n_G(x,A_2)] \geq 1/2$ for all $y \in B \cap V_2$.
\end{proof}

\begin{lemma}\label{lem:min-oneset}
Let $G=(V,E) \sim DSBM(V_1,V_2,p,q)$ where $p,q,$ and $n=|V|$ satisfy the conditions in the statement of Theorem~\ref{thm:main-accuracy-dsbm}.
Let $S_a,S_b$ be two subsets considered in one of Phase 1's recursive calls to Update during the execution of Algorithm~\ref{alg:dir_sbm} on $G$, and let $S_{b,1},S_{b,2}$ be the subsets returned by $\mathrm{Update}$. Assume that $|S_{a,1}|,|S_{a,1}| > 0$. Then with probability at least $1-n^{99}$, it holds that $|S_{b,1}| \geq \frac{n}{18\sqrt{\log n}}$ and $|S_{b,2}| \geq \frac{n}{18\sqrt{\log n}}$.
\end{lemma}

\begin{proof}
First, note that in  Algorithm~\ref{alg:dir_sbm}, $\mathrm{Update}$ is called not on $G$ but on $H=\edgeflip_{\po}(G)$, where $\po = \min\left(\frac{96 \log(2/\delta)}{\eps^2 \ell},1/2\right)$ and $\ell = \frac{n}{18\sqrt{log n}}$. By Lemma~\ref{lem:noised-dsbm}, we have that $H \sim DSBM(V_1,V_2,p',q')$ for $p'=p(1-2\po) + \po$ and $q'=q(1-2\po) + \po$. In particular, we still have $p' > q'$, and thus we can apply Lemma~\ref{lem:preference} with $A=S_a$ and $B=S_b$. 

Without loss of generality, suppose that the first bullet of the lemma holds, i.e., for all $x \in S_b \cap V_1$, $P[(n_H(x,S_{a,1}) \ge n_H(x,S_{a,2})] \geq 1/2$ and for all $y \in S_b \cap V_2$, $P[(n_H(y,S_{a,2}) \ge n_H(y,S_{a,1})] \geq 1/2$. 
In the case that $n_H(x,S_{a,1}) = n_H(x,S_{a,2})$, $x$ is placed in to either of $S_{b,1}, S_{b,2}$ with probability $1/2$. In the case that $n_H(x,S_{a,1}) > n_H(x,S_{a,2})$, $x$ is placed into $S_{b,1}$ with probability $1$. Thus $P[x \in S_{b,1}] \geq 1/4$. 

For each $x \in V_1 \cap S_B$, let $A_x$ be the random variable for the event that at the end of this call to $\mathrm{Update}$, $x \in S_{b,1}$. By the above, $\E[A_x] \geq 1/4$, and 
\[ \E[|S_{b,1}|] \geq \sum_{x \in V_1 \cap S_B} A_x \geq \frac{1}{4}\cdot\frac{n}{4\sqrt{\log n}}.\] 
Since the edges from $x_1$ to $S_{a}$ and from $x_2$ to $S_a$ are independent for any $x_1 \neq x_2$, we can apply the standard Chernoff bound (\cite{book_prob_MU05}, Theorems 4.4 and 4.5) to say that 
\[ P\left[|S_{b,1}| \geq \frac{n}{18\sqrt{\log n}} \right] \geq 1 - \exp(-\E[S_{b,1}](1-8/9)^2/3) = 1 - \exp\left( \frac{-n}{243\cdot8\sqrt{\log n}\cdot}\right) \geq 1 - n^{-100}.\] 

By an identical argument, $P\left[|S_{b,2}| \geq \frac{n}{18\sqrt{\log n}} \right] \geq 1 - n^{-100}$,
so the lemma follows by a union bound on these two sets.
\end{proof}

\begin{proof}[Proof of \Cref{lem:update-minsize}]
To conclude the proof, 
we apply a union bound on Lemma~\ref{lem:min-oneset} over all calls to $\mathrm{Update}$. 

There remains only to show that every time we call $\mathrm{Update}(H,\ell,S_b,S_{a,1},S_{a,2})$, we have $|S_{a,1}| = |S_{a,2}| > 0$, as this is a condition in the statement of Lemma~\ref{lem:preference}.
Both conditions are guaranteed deterministically by $\mathrm{Update}$. 
The equality of the subset sizes follows from the subsampling procedure in lines 10 -- 16 of the pseudocode. 
The fact that both sets used in $\mathrm{Update}$ are nonempty is also true deterministically -- either this condition holds already (which will be true with very high probability), or it is enforced by lines 2 -- 9. 
\end{proof}

We can now prove Theorem~\ref{thm:main-accuracy-dsbm}.

\begingroup
\def\thetheorem{\ref{thm:main-accuracy-dsbm}}
\begin{theorem}
Let $V$ be a vertex set of size $n$ with planted partition $(V_1,V_2)$, and let $\delta, \eps > 0$.
Let $p,q>0$ such that $\frac{p-q}{\sqrt{p}} \in \Omega\left(\frac{\sqrt{\log n}}{\sqrt{n}}\right)$ and $p \in \Omega\left( \frac{\log(1/\delta)\sqrt{\log n}}{\eps^2 n}\right)$,
and suppose $n \in \Omega\left( \frac{\log(1/\delta)\sqrt{\log n}}{\eps^2 } \right)$.
Then 
with probability $1-1/poly(n)$ over the draw  from $DSBM(V_1,V_2,p,q)$ and the edge flipping  by $\edgeflip_{\po}$, we get a graph $H$
such that
Algorithm~\ref{alg:dir_sbm} recovers the underlying communities from $H$
with probability at least $1/3$ over the algorithm's partitioning randomness (and independently from the randomness 
of $DSBM(V_1,V_2,p,q)$ and the edge-flipping randomness).
Moreover, for $\eps \in \Omega\left (\frac{1}{\poly(\log(n))}\right), \delta \in  \Omega \left(\frac{1}{\poly(n)}\right)$  the algorithm uses time and space $\widetilde{O}(|V|+|E|)$.
\end{theorem}
\endgroup

\begin{proof}
By Lemma~\ref{lem:update-minsize}, with probability at least $1-1/n^{100}$ over the draw from $DSBM(V_1,V_2,p,q)$ and the edge flipping, no subset considered by a call to $\mathrm{Update}$ is ever below the size threshold of $\ell$.
Therefore, with this probability, lines 2 -- 9 of the pseudocode of $\mathrm{Update}$ are never invoked.
In that case, Algorithms~\ref{alg:dir_sbm} and Algorithms~\ref{alg:dir_int} are identical, and we can apply Lemma~\ref{lem:int-accuracy}, which holds with probability at least $1-poly(n)$ over the draw from $DSBM(V_1,V_2,p,q)$ and the edge flipping. A union bound over these two lemmas concludes the proof. 
\end{proof}

\subsection{Degree-Heterogeneous SBM} \label{sec:dhsbm}
We note that our Algorithm~\ref{alg:dir_sbm} makes no assumption that the edge probabilities are homogeneous across nodes (or indeed, any assumptions on these probabilities at all, as they are unknown to the algorithm).
Therefore, it can be run on general input graphs, including degree-heterogeneous SBM graphs, and our unconditional privacy result in Theorem~\ref{thm:dsbm-privacy} extends directly. 
For DHSBM input graphs, our utility result applies with the minor modification of applying the conditions on $p,q$ to all node-specific probabilities $p_u,q_u$.

More formally, for a vertex set $V = V_1 \cup V_2$, let $\mathcal{PQ} = \{(p_u,q_u) | u \in V\} \subset [0,1] \times [0,1]$ denote a set of tuples of probabilities. We write $DHSBM(V_1,V_2,\mathcal{PQ})$ for the DHSBM on $V_1 \cup V_2$ where every node $u$ has an edge to each node in its home community with probability $p_u$ and to each node in the opposite community with probability $q_u$. 

The proof of utility in the DHSBM setting follows by an analagous argument to the proof of Theorem~\ref{thm:main-accuracy-dsbm}.

\begin{theorem}\label{thm:main-accuracy-dhsbm}
Let $V$ be a vertex set of size $n$ with planted partition $(V_1,V_2)$. Let $\mathcal{PQ} = \{(p_u,q_u) | u \in V \}$ be a set of tuples in $[0,1] \times [0,1]$ such that for all $u\in V$, $p_u,q_u$ satisfy $\frac{p_u-q_u}{\sqrt{p_u}} \in \Omega\left(\frac{\sqrt{\log n}}{\sqrt{n}}\right)$ and $p_u \in \Omega\left( \frac{\log(1/\delta)\sqrt{\log n}}{\eps^2 n}\right)$. Let $\delta, \eps > 0$ and suppose $n \in \Omega\left( \frac{\log(1/\delta)\sqrt{\log n}}{\eps^2 } \right)$.
Then 
with probability $1-1/poly(n)$ over the draw  from $DHSBM(V_1,V_2,\mathcal{PQ})$ and the edge flipping  by $\edgeflip_{\po}$, we get a graph $H$
such that
Algorithm~\ref{alg:dir_sbm} recovers the underlying communities from $H$
with probability at least $1/3$ over the algorithm's partitioning randomness (and independently from the randomness 
of $DHSBM(V_1,V_2,\mathcal{PQ})$ and the edge-flipping randomness).
Moreover, for $\eps \in \Omega\left (\frac{1}{\poly(\log(n))}\right), \delta \in  \Omega \left(\frac{1}{\poly(n)}\right)$  the algorithm uses time and space $\widetilde{O}(|V|+|E|)$.
\end{theorem}

\begin{proof}
We first need an analogous lemma to Lemma~\ref{lem:noised-dsbm}:

\begin{lemma}\label{lem:noised-dhsbm}
Let $V$ be a vertex set with planted partition $(V_1,V_2)$, and let $\mathcal{PQ} = \{(p_u,q_u) | u \in V \}$ be a set of tuples in $[0,1] \times [0,1]$.
Let $\po \in [0,1]$.
For any graph $G=(V,E) \sim DHSBM(V_1,V_2,\mathcal{PQ})$, the edge-flipped graph $H =\edgeflip_{\po}(G)$ is distributed according to $DHSBM(V_1,V_2,\mathcal{PQ}')$, where for all $u \in V$, the tuples $(p_u,q_u) \in \mathcal{PQ}$, $(p_u',q_u')\in \mathcal{PQ'}$ satisfy $p_u'=p_u(1-2\po) + \po$ and $q_u'=q_u(1-2\po) + \po$.
\end{lemma}

\begin{proof}
   Repeat the argument from the proof of Lemma~\ref{lem:noised-dsbm} with $p,q$ replaced with $p_u,q_u$ for every $u \in V$.
\end{proof}

The theorem follows by repeating the argument from the proof of Theorem~\ref{thm:main-accuracy-dsbm} by again replacing $p,q$ with $p_u,q_u$ for all $u \in V$, and applying Lemma~\ref{lem:noised-dhsbm} in place of Lemma~\ref{lem:noised-dsbm}.
\end{proof}

\section{Undirected Stochastic Block Model}\label{sect:SBM}
`In this section, we extend our results for the directed SBM to the standard (undirected) SBM. 
Our algorithm for this model is presented in Algorithm~\ref{alg:undir_sbm}, where the subroutines  $\mathrm{Phase 1}$ and $ \mathrm{Update}$ are the same as in the prior section.

A natural approach would be to apply a reduction, as in \cite{colt22directed}, from the SBM to the DSBM by replacing each undirected edge with two bidirectional edges, and then applying the directed algorithm. However, this solution would be insufficient for obtaining a differentially private algorithm, as some edges would be examined twice in determining the communities. We instead apply our private edge flipping mechanism twice, once in the beginning of the algorithm, and once in independently reflipping the edges in between the two executions of $\mathrm{Update}$.

\begin{algorithm}
\noindent \textbf{Input:} A graph $G = (V, E)$ with $n$ vertices, privacy parameters $\eps,\delta$

\noindent \textbf{Output:} A bipartition of $V$

\begin{algorithmic}[1]
\STATE $\ell \gets \frac{n}{18\sqrt{\log n}}$; $\po \gets \min\left(\frac{1536 \log(8/\delta)}{\eps^2 \ell},1/2\right)$
\STATE Graph $H \gets \edgeflip_{\po}(G)$

\STATE $S,S' \gets$ a random partition of $V$ into two disjoint subsets of size $n/2$
\STATE $H_S \gets$ the subgraph of $H$ induced by $S$; $H_{S'} \gets$ the subgraph of $H$ induced by $S'$

\STATE $S_1,S_2 \gets \mathrm{Phase 1}(H_S,\ell)$
\STATE $S'_1,S'_2 \gets \mathrm{Phase 1}(H_{S'},\ell)$
\STATE $S^*_1,S^*_2 \gets \mathrm{Update}(H,\ell,S,S'_1,S'_2)$

\STATE Graph $\widetilde{H} \gets \edgeflip_{\po}(G)$
\STATE $S^{'*}_1,S^{'*}_2 \gets \mathrm{Update}(\widetilde{H},\ell,S',S^*_1,S^*_2)$
\RETURN $(S^*_1 \cup S^{'*}_1, S^*_2 \cup S^{'*}_2)$

\caption{$\mathrm{PrivateSBMRecovery}(G,\eps,\delta)$ an algorithm for exact recovery in the SBM} \label{alg:undir_sbm}
\end{algorithmic}
\end{algorithm}

Thanks to our framework, it is easy to show that this algorithm ensures privacy. Showing that it still recovers the communities correctly, however, requires some care, as the algorithm now operates on two (partially) dependent SBM graphs. We will use a careful coupling argument between the DSBM and the distribution of the intermediate directed graphs used in the execution of our undirected algorithm, which allows us to apply utility properties of the directed algorithm. 

\begin{theorem}\label{thm:undir_dp}
For any input graph $G$, Algorithm~\ref{alg:undir_sbm} is $(\eps,\delta)$-differentially private. 
\end{theorem}

\begin{proof}
We break up the execution of Algorithm~\ref{alg:undir_sbm} into two pieces: Algorithm 3.1, which runs from the beginning to the first call to $\mathrm{Update}$ (lines 1 -- 8 of the pseudocode), and Algorithm 3.2, which consists of the remainder of the algorithm (the second call to the edge-flipping mechanism $\edgeflip_{\po}(G)$ and the last call to $\mathrm{Update}$). Algorithm 3.1 satisfies Definition~\ref{def:disjoint-star-ada} of an adaptive disjoint-star algorithm with $\ell = \frac{n}{18\sqrt{\log n}}$, so it is $(\eps/4,\delta/4)$-DP by Theorem~\ref{thm:privacy-adptive-stars} and the choice of the parameter $\po$. Algorithm 3.2 is a degree-counting-in-disjoint-stars instance (Definition~\ref{thm:privacy-adptive-stars}, and thus is similarly $(\eps/4,\delta/4)$-DP by Theorem~\ref{thm:privacy-adptive-stars}. 
Therefore, Algorithm 3.2 is $(\eps/2,\delta/2)$-DP by composition.
Therefore, the full Algorithm~\ref{alg:undir_sbm} is $(\eps,\delta)$-DP by group privacy. 
\end{proof}

\subsection{Proof of Correctness}\label{app:undir-utility}
As previously discussed, we prove correctness of our undirected algorithm using a coupling argument between intermediate directed graphs it constructs and graphs drawn from the DSBM.
Specifically, we couple the graph accessed in the first portion of Algorithm~\ref{alg:undir_sbm} with that of Algorithm~\ref{alg:dir_sbm}, ending at the first call of each algorithm to the $\mathrm{Update}$ subroutine (lines 1 -- 8 of the pseudocode in both Algorithm~\ref{alg:undir_sbm} and Algorithm~\ref{alg:dir_sbm}). 

The key insight is that the first portion of Algorithm~\ref{alg:undir_sbm} scans the graph by looking at edges in a specific direction that is independent of the graph structure and only depends on the randomness of the algorithm. More precisely, during the execution of $\mathrm{Phase 1}$ and the first call of $\mathrm{Update}$   the algorithm partitions the nodes into random subsets $S_1,\dots, S_k$ and then checks the degrees of  nodes in one subset $S_i$ to (a subset of) some other subset $S_j$ (for different $i$'s and $j$'s). Interestingly, for any pair of sets $S_i$ and $S_j$, $\mathrm{Phase 1}$ either analyzes the degrees of nodes in $S_i$ to subsets in $S_j$ or from nodes in $S_j$ to subsets in $S_i$ but not both.
So we only look at one direction of each edge, and this direction depends only on the random partition and not
on the graph structure. 
Consequently, if we couple the 
random partitions of Algorithm~\ref{alg:dir_sbm} 
on $DSBM(V_1,V_2,p,q)$ and  the random partition of Algorithm~\ref{alg:undir_sbm}
on $SBM(V_1,V_2,p,q)$,
they both look at exactly the same (directed) edges during $\mathrm{Phase 1}$ and the first call to $\mathrm{Update}$, and fail with the same probability. Now we are ready to state our Lemma.

\begin{lemma}\label{lem:coupling}
Let $V$ be a vertex set of size $n$ with planted partition $(V_1,V_2)$, and let $\eps,\delta > 0$. Let $p,q>0$ such that $\frac{p-q}{\sqrt{p}} \in \Omega\left(\frac{\sqrt{\log n}}{\sqrt{n}}\right)$.
Then with probability at least $1 - 2 \exp(-100 \log n)$ over the draw from $SBM(V_1,V_2,p,q)$ 
and the edge flipping  by $\edgeflip_{\po}$, we get a graph $H$ such that
at the end of the first phase of the execution of Algorithm~\ref{alg:undir_sbm},
the obtained partition $P^* = (S_1^*,S_2^*)$ satisfies either $S_1^*\subseteq V_1$ and $S_2^*\subseteq V_2$ or $S_1^*\subseteq V_2$ and $S_2^*\subseteq V_1$ with probability at least $1-4/160$ over the algorithm's partitioning randomness (and independently from the randomness of $SBM(V_1,V_2,p,q)$ and the edge-flipping randomness).
\end{lemma}

\begin{proof}
From a high level perspective, we show that the during Phase 1, Algorithm~\ref{alg:undir_sbm},  behaves as Phase 1 of the execution of Algorithm~\ref{alg:dir_sbm} on a graph $G'=(V,E')\sim DSBM(V_1,V_2,p',q')$ with parameters $\eps/4$ and $\delta/4$. So the lemma follows by applying Lemma 12 from~\cite{colt22directed}.

Consider the partitions $S$ and $S'$ computed by Algorithm~\ref{alg:dir_sbm}
and Algorithm~\ref{alg:undir_sbm}.
These partitions are independent of the graph structure. We couple the choice of Algorithm~\ref{alg:undir_sbm} to the choice of Algorithm~\ref{alg:dir_sbm} so they draw the same partition. Similarly, during Phase 1, the sets $S$ and $S'$ are partitioned randomly in $\sqrt{\log n}$ sets by the two algorithms, independently of the graph structure. 
We also couple these partitions so the two algorithm use the same one.

Now there are two important things to notice. During Phase 1,  Algorithm~\ref{alg:undir_sbm} checks the presence of a particular  edge at most once. In addition, the algorithm looks at only one arc 
out of each pair of antiparallel arcs (that replace an undirected edge). The one that it looks at
 depends only on the random partitions.

We continue the coupling of
 the execution of the Phase 1 of Algorithm~\ref{alg:undir_sbm} on $G=(V,E)\sim SBM(V_1,V_2,p,q)$ to  a run of Phase 1 of Algorithm~\ref{alg:dir_sbm} on  $G=(V,E)\sim DSBM(V_1,V_2,p,q)$ with parameters $\eps/4$ and $\delta/4$. So far we coupled the partitions that they use. 
Now we  note that both algorithms make each of their choices based on the degree of a particular node $u$ to a particular set $B$. In order to compute this degree we have to sample the edges of $u$ to the set $B$ in $H=(V,E)$ and $H'=(V,E')$, where $H$ and $H'$ are the graphs obtained after running edge-flipping on $G$ and $G'$. We do this by coupling the draw of an edge $(u,v)$ in $H(V,E)$ to the draw of the directed  edge from $u$ to $v$ in $H'(V, E')$. (So we have an edge $(u,v)$ in $H$ if and only if we have a directed edge from $u$ to $v$ in $H'$.) Note that we can do this because during the execution of Phase 1 we materialize the edge just once as observed before, and because the probability that the edge exists in the two graphs is the same (this is because $p$ and $q$ are the same in both models and $\po$ is the same by the choice of privacy parameters by the two algorithms). 

It follows from our coupling that the two  algorithms behave identically during Phase 1 and in particular output the same partitions with the same probability. Thus the guarantee of Lemma 12 from~\cite{colt22directed} on the output of Phase 1 of Algorithm~\ref{alg:dir_sbm} on $G'=(V,E')\sim DSBM(V_1,V_2,p',q')$ also holds for the output of Phase 1 of Algorithm~\ref{alg:undir_sbm} on $G=(V,E)\sim SBM(V_1,V_2,p',q')$. Thus the lemma follows.
\end{proof}

To conclude the proof of correctness, we need to show that also the second call to the $\mathrm{Update}$ subroutine returns a perfect partition. The proof of the following lemma follows from two applications of a Chernoff bound on negatively correlated and independent variables.

\begin{lemma}\label{lem:vertex_assignment}
Let $v$ be a node in $S'\cap V_1$ (resp. $S'\cap V_2$) and let $C=\min\{|V_1\cap S|, |V_2\cap S|\}$. Let $p,q>0$ such that $\frac{p-q}{\sqrt{p}} \in \Omega\left(\frac{\sqrt{\log n}}{\sqrt{n}}\right)$. 
Then with probability at least $1- 8 n^{-3}$ over the draw from $SBM(V_1,V_2,p,q)$ 
and the edge flipping  by $\edgeflip_{\po}$, we get a graph $\widetilde{H}$ such that
$v$ has more edges in $\widetilde{H}$ 
to a sample of size $C$ of $V_1\cap S$ (resp. $V_2\cap S$) than to a sample of size $C$ of $V_2\cap S$ (resp. $V_1\cap S$) with probability at least $1-8n^{-4}$ over the algorithm's partitioning randomness (and independently from the randomness of $SBM(V_1,V_2,p,q)$ and the edge-flipping randomness).
\end{lemma}

\begin{proof}
Note that the set $S$ and $S'$ are randomly selected subsets of size $n/2$. 
    
As a first step we want to lower bound the probability that $|V_1\cap S|$ or $|V_2\cap S|$ are smaller than $n/4-2\sqrt{n\log n}$ with $4n^{-4}$. Let $X_i$ be 1 if the $i$-th element in $S$ is in $V_1$ and $0$ otherwise. Note that $X_i=1$ with probability $1/2$ and that the variables $X_i$ are negatively correlated, and that $\sum_{i=1}^{\nicefrac n2}X_i$ is $|V_1\cap S|$. Now $E[\sum_{i=1}^{\nicefrac n2}X_i]=n/4$. So by applying the Chernoff bound for negatively correlated random variables we have that $\Pr[|\sum_{i=1}^{\nicefrac n2}X_i - E[\sum_{i=1}^{\nicefrac n2}X_i]|\geq 2\sqrt{n\log n}]\leq 2n^{-4}$. So the probability that both $|V_1\cap S|$ and $|V_2\cap S|$ are bigger than $n/4-2\sqrt{n\log n}$ is $1-8n^{-4}$.

Now, let $v$ be a node in $S'\cap V_1$. Let $deg_{V_1\cap S}(v)$ be the number of edges between $v$ and a sample of size $C$ of $V_1\cap S$ and define $deg_{V_2\cap S}(v)$ symmetrically. We have that $E[deg_{V_1\cap S}(v)] = C  p'$ and similarly $E[deg_{V_2\cap S}(v)] = C q'$. Furthermore, the variables are independent, so by a Chernoff bound we have that $\Pr[|deg_{V_1\cap S}(v) - E[deg_{V_1\cap S}(v)]| \geq 3\sqrt{C p' \log n}]\leq 2 n^{-3}$ and similarly $\Pr[|deg_{V_2\cap S}(v) - E[deg_{V_2\cap S}(v)]| \geq 3\sqrt{C q' \log n}]\leq 2 n^{-3}$. So with probability $1- 8 n^{-3}$, $deg_{V_1\cap S}(v)\geq C  p' - 3\sqrt{C p' \log n}$ and $deg_{V_2\cap S}(v)\leq C  q' + 3\sqrt{C q' \log n}$ so $deg_{V_1\cap S}(v)- deg_{V_2\cap S}(v)\geq C(p'-q') - 3\sqrt{C p' \log n}- 3\sqrt{C q' \log n}\geq 40 C \sqrt{p'}\frac{\sqrt{\log n}}{\sqrt{n}} - 6\sqrt{C p' \log n} $ by selecting the constant hidden in the $\Omega$ notation to be such that  $\frac{p'-q'}{\sqrt{p'}} \geq 40\left(\frac{\sqrt{\log n}}{\sqrt{n}}\right)$. Now using the fact that $C$ is bigger $n/4-2\sqrt{n\log n}$ with probability $1-4n^{-4}$ on the algorithm's partitioning randomness, we get that $deg_{V_1\cap S}(v)- deg_{V_2\cap S}(v)>0$ with probability $1-8n^{-4}$ on the algorithm's partitioning randomness and $1- 4 n^{-3}$ on the graph's randomness.

The lemma follows by applying an identical argument to nodes $v\in S'\cap V_2$.
\end{proof}

We are now ready to prove the main SBM theorem.

\begin{theorem}\label{thm:sbm-main}
Let $V$ be a vertex set of size $n$ with planted partition $(V_1,V_2)$, and let $\eps,\delta > 0$. 
Let $p,q>0$ such that 
$\frac{p-q}{\sqrt{p}} \in \Omega\left(\frac{\sqrt{\log n}}{\sqrt{n}}\right)$ and $p \in \Omega\left( \frac{\log(1/\delta)\sqrt{\log n}}{\eps^2 n}\right)$,
 and suppose $n \in \Omega\left(\frac{\log(1/\delta)\sqrt{\log n}}{\eps^2 }\right)$.
Then 
with probability $1-1/poly(n)$ over the draw from $SBM(V_1,V_2,p,q)$ and the edge flippings by the two calls to $\edgeflip_{\po}$, we get a pair of graphs $H$ and $\widetilde{H}$ 
such that Algorithm~\ref{alg:undir_sbm} recovers the underlying communities from $H$ and $\widetilde{H}$ with probability at least $1/3$ over the algorithm's partitioning randomness (independently from the randomness 
of $SBM(V_1,V_2,p,q)$ and the edge-flipping randomness).
Moreover, for $\eps \in \Omega(1/\poly(\log(n))), \delta \in  \Omega(1/\poly(n))$  the algorithm uses time and space $\widetilde{O}(|V|+|E|)$.
\end{theorem}

\begin{proof}
First note that by the assumptions on $p$, $q$, and $n$ we have that $\frac{p-q}{\sqrt{p}} \in \Omega\left(\frac{\sqrt{\log n}}{\sqrt{n}}\right)$ implies $\frac{p'-q'}{\sqrt{p'}} \in \Omega\left(\frac{\sqrt{\log n}}{\sqrt{n}}\right)$. So by Lemma~\ref{lem:coupling}, we know that all nodes are classified correctly by the first call to the $\mathrm{Update}$ subroutine. Thus by  Lemma~\ref{lem:vertex_assignment} and a union bound over all vertices and the two phases, we have that all nodes are classified correctly with the desired probabilities. The running time follows from Lemma~\ref{lemma:edge-flipping-time}, our choice of parameters, and the fact that $\mathrm{Phase 1}$ and $\mathrm{Update}$ perform linear scans over the edges in the graphs.
\end{proof}
\section{Privately boosting the success probability}\label{sec:boosting}
In previous sections, we showed that it is possible to obtain a private algorithm that succeeds with constant probability in correctly outputting the planted partition. It is trivial to decrease the probability of failure to a $1/\mathrm{poly}(n)$ by $O(\log(n))$ parallel repetitions of the algorithm. However, parallel repetitions result in an undesirable $\Omega(\log(n))$ factor increase in privacy loss.  Notice that requring $\eps$ that is $\Omega(\log(n))$ larger, would induce a  $\Omega(\log^2(n))$ factor increase in the threshold for parameter $p$ of the model that is required to achieve exact recovery. This is undesirable as it would further move the algorithm away from the information-theoretic threshold. For this reason in this section we use recent results on boosting with privacy amplification~\cite{cohen2022generalized} to circumvent that problem. We can thus show how to boost the probability of success to $1-O(1/n)$ at the only expense of a constant factor increase in the privacy parameter and a small sublinear factor increase in running time. 

We first introduce some notation to define precisely the properties required for this amplification. We define this notation in terms of general classes of SBM graphs (directed or undirected) to prove amplification for both our algorithms.

Let $\mathcal{D}$ be a graph distribution (e.g., the directed DSBM with some fixed parameters $V_1, V_2, p,q$). For this subsection, we assume that we have a randomized (DP) algorithm $\mathcal{A}$ that given a graph $G \sim \mathcal{D}$ outputs a solution $\mathcal{A}(G)=(S_1,S_2)$. We also assume that we are given a (randomized) DP verifier algorithm $\mathcal{V}$ that given a solution $(S_1,S_2)$ and the graph $G$ outputs a score $\mathcal{V}(S_1,S_2,G)$, where a lower score indicates a (more) correct solution. In addition, we  assume that there is a threshold $\theta_{\mathcal{D}}$ (depending on the distribution and not necessarily known to the algorithm) such that the algorithm outputs a score below $\theta_{\mathcal{D}}$ for the correct solution and above $\theta_{\mathcal{D}}$ for {\emph all} incorrect solutions with high probability. Notice that the verifier algorithm needs to be DP with respect to the input graph only (not with respect to the solution) and needs only to be correct with high probability.

For simplicity we show the result for the $DSBM$; an identical proof also applies to the $SBM$. Given a solver $\mathcal A$, we define a graph $G \sim DSBM(V_1, V_2, p,q)$ to be $(q,\theta_{\mathcal D})$-\emph{suitable} with respect to $\mathcal A$
if it satisfies the following two conditions:
(1) the probability that $\mathcal A$ returns the correct partition on $G$ is at least $q$; and (2) $\Pr \{\mathcal{V}(S_1,S_2,G) > \theta_{\mathcal{D}} \} \ge 1 -o(1/n^2)$ and for any $(S_1,S_2)\neq (V_1,V_2)$, $\Pr \{\mathcal{V}(V_1,V_2,G) \le \theta_{\mathcal{D}} \}\ge 1- o(1/n^2)$. 
Given these two algorithms, we can prove the following main theorem (which is an application of the framework of~\cite{cohen2022generalized}).

\begin{theorem}\label{thm:boosting}
  Fix $\eps, \delta, \gamma>0$. 
  Fix $\tau >0$, and $\mathcal A$ be an $(\eps, \delta)$-DP solver algorithm such that there exists a threshold $\theta_{\mathcal{D}}$ for
  which 
  $$\Pr_{G \sim DSBM(V_1,V_2, p,q)} \{G \text{ is $(q,\theta_{\mathcal D})$-suitable w.r.t. $\mathcal A$} \}\ge 1-o(1/n^2).$$
Then there exists a $((4+2\gamma)\eps, 2n^{O(1/\gamma)}\log(n)\delta)$-DP algorithm that outputs the correct solution with probability $1-o(1/n)$ and makes only $O(q^{-1} n^{O(1/\gamma)}\log(n))$ calls to the solver and verifier algorithms.
\end{theorem}

Before proving the theorem above we need to show some technical lemmas.
To prove Theorem~\ref{thm:boosting}, we review a recent result from~\cite{cohen2022generalized}. We define an algorithm that executes multiple times a combination of a solver algorithm $\mathcal{A}$ and a verifier algorithm $\mathcal{V}$. This algorithm assumes an arbitrary total order over the set of all bipartitions  (e.g., based on a lexicographic order of the nodes).

\begin{algorithm}
	\begin{algorithmic}[1]
\REQUIRE A graph $G = (V, E)$ with $n$ vertices and $\eps, \delta, \gamma, \tau > 0$ a solver algorithm $\mathcal{A}$ and a verifier algorithm $\mathcal{V}$. 
\caption{$\mathrm{SelectBest}(G, \eps, \delta, \tau, \gamma, \mathcal{A}, \mathcal{V} )$ -- Selecting the best solution among various trials
		\label{algo:select}}
\STATE Sample $p$ from $[0,1]$ where $\Pr(p\le x) =  x^\gamma$ for all $x\in [0,1]$.

\STATE $C= \emptyset$
\FOR{$i \in [\tau]$} 
\IF{ $\bern{p} = 1$ }
\STATE $S_1,S_2 = \mathcal{A}(G, \eps,  \delta)$.
\STATE $s = \mathcal{V}(G, S_1, S_2, \eps,  \delta)$.
\STATE $C \gets C \cup (s,S_1,S_2)$
\ENDIF
\ENDFOR
\RETURN Solution in $C$ with lowest $s$ (ties broken by an arbitrary total order on $S_1,S_2$). 
\end{algorithmic}
\end{algorithm}
The following lemma proves that the algorithm is private and is a corollary of Theorem 1 in~\cite{cohen2022generalized}.

\begin{lemma}\label{lemma:many-runs-private}
Given $\eps, \delta, \gamma,  \tau > 0$, 
Algorithm $\mathrm{SelectBest}(G, \eps, \delta, \tau, \gamma, \mathcal{A}, \mathcal{V})$ is $((4+2\gamma) \eps, 2\tau \delta)$-DP and the algorithm makes at most $\tau$ calls to $\mathcal{A},\mathcal{V}$.
\end{lemma}
\begin{proof}
First notice that using $\mathcal{A}, \mathcal{V}$ we can obtain a $(2\eps,2\delta)$-DP algorithm that outputs $(s,S_1,S_2)$ with $S_1,S_2 = \mathcal{A}(G, \eps,  \delta)$, $s = \mathcal{V}(G, S_1, S_2, \eps,  \delta)$. The outputs of this algorithm are ranked by a total order given by the score in increasing order and with ties broken by the arbitrary total order of the partitions. We can then invoke Theorem 1 in~\cite{cohen2022generalized} to show that Algorithm $\mathrm{SelectBest}(G, \eps, \delta, \tau, \gamma, \mathcal{A}, \mathcal{V})$ is $((2+\gamma)2\eps,  2\tau\delta)$-DP. The bound on the number of calls (and hence running time) is trivial.
\end{proof}

We now prove the correctness of the algorithm. This result is similar to prior work~\citep{epasto2024power}.

\begin{lemma}\label{lemma:many-runs-correct}
  Fix $\tau, \gamma>0$, let $\mathcal{D}$ be a graph distribution, 
 let $\mathcal A$ be a solver algorithm such that there exists a threshold $\theta_{\mathcal{D}}$ for
  which 
  $$\Pr_{G \sim \mathcal{D}} \{G \text{ is $(q,\theta_{\mathcal D})$-suitable w.r.t. $\mathcal A$} \}\ge 1-o(1/n^2).$$
Then  $\mathrm{SelectBest}(G, \eps, \delta, \tau, \gamma, \mathcal{A}, \mathcal{V})$ with $\tau \in O(q^{-1}\log(n)n^{O(1/\gamma)})$ outputs the correct solution with probability $1-o(1/n)$.
\end{lemma}
\begin{proof}
We assume that the event $G\in \mathcal{G}$ is realized. (This event happens with high probability because of the assumption of the lemma). By the definition of the algorithm, let $c$ be a large enough constant such that, $\Pr(p\le 1/n^{c/\gamma}) =  n^{-c} \in o(1/n^2)$. For the rest of the proof we condition on $p> 1/n^{c/\gamma}$. Conditioning on this event and using a Chernoff bound, 1 is sampled by the Bernoulli variable 
$\Theta(\tau /n^{-c/\gamma})$ times  with probability $1-O(\exp (\tau /n^{-c/\gamma}))$. Choosing $\tau \in \Theta(1/n^{c/\gamma} \log(n) q^{-1})$, with  probability $1-o(1/n^2)$, we have $\Theta(\log(n) q^{-1})$ calls to the solver and verifier algorithm.
Since each call to the verifier algorithm succeeds with probability $q$, the probability that all fail is $(1-q)^{\Theta(\log(n) q^{-1})} \in o(1/n)$. So we can assume that at least one succeeds with probability $1-o(1/n^2)$. Finally, notice that all of at most $\tau$ calls of the verifier output a value below $\theta_{\mathcal{D}}$ for any correct solution and a value above $\theta_{\mathcal{D}}$ for any incorrect solution,  with probability at least $1-o(1/n)$.
Combining all low probability events of failures, we have that the algorithm succeeds with probability $1-o(1/n)$ because one correct solution is ranked first (as it has the lowest value). 
\end{proof}

We are now ready to prove Theorem~\ref{thm:boosting}.

\begin{proof}[Proof of Theorem~\ref{thm:boosting}]
  We analyze Algorithm~\ref{algo:select} combined with our solver from \Cref{sec:dp-private} (analyzed in the proof of Theorem~\ref{thm:main-accuracy-dsbm}).
  The privacy follows directly from Lemma~\ref{lemma:many-runs-private}, Theorem~\ref{thm:boosting:private}
  and Theorem~\ref{thm:main-accuracy-dsbm}.

  Thus, we need to show that the correctness holds. We analyze the probability distribution over graphs obtained
  by sampling from the $DSBM(V_1,V_2, p,q)$ and applying the edge flipping algorithm. An immediate coupling
  shows that this distribution, call it $\mathcal D$, follows $DSBM(V_1,V_2, p',q')$.
  To do so, we are going to apply  Lemma~\ref{lemma:many-runs-correct} and
  we thus first argue that its requirements are satisfied with the probability distribution $\mathcal D$.
  First, observe
  that by Theorem~\ref{thm:main-accuracy-dsbm} the probability
  that the graph generated from the input graph $G$ and applying
  the edge flipping algorithm belongs to the set $\mathcal{G}$ 
  defined in Theorem~\ref{thm:main-accuracy-dsbm} is at least
  $1-\tau n^{-\Omega(1)}$. Therefore, an immediate union bound 
  yields that all the $\tau$ graphs generated by Algorithm~\ref{algo:select}
  belongs to $\mathcal{G}$ with 
  probability at least $1-\tau n^{-\Omega(1)}$ independently 
  from the randomness of the solver. Moreover,
  Theorem~\ref{thm:verifier} succeeds with probability 
  at least $1-n^{-\Omega(1)}$ independently from the randomness
  of the solver as well. Therefore a union bound implies that with probability
  at least $1-2\tau n^{-\Omega(1)}$, the graphs generated belong to
  $\mathcal{G}$ as defined in Lemma~\ref{lemma:many-runs-correct}. Therefore
  the algorithm succeeds with probability at least $1-O(\tau n^{-\Omega(1)})$, as desired.
\end{proof}

The existence of DP solver algorithms for regimes of the DSBM and SBM model has been proven in the previous sections. In the next section, we show how to design a DP verifier using our framework to get our results. 

\subsection{Private verification of a (D)SBM solution}\label{sec:verifier}
To complete the proof we need to show that there exists a DP verifier algorithm that is correct with high probability. We do so in this section, showing that such an algorithm can be obtained using our framework.

We define the following algorithm for directed graphs (a similar one can be given for undirected graphs).
\begin{algorithm}
	\begin{algorithmic}[1]
\REQUIRE A directed graph $G = (V, E)$ with $n$ vertices and a partition $S_1, S_2$, $\eps, \delta > 0$.
\IF  {$|S_1| \neq n/2$}
\RETURN $\infty$
\ENDIF
\STATE $\ell \gets \frac{n}{2}$, $\po \gets \min\left(\frac{96 \log(2/\delta)}{\eps^2 \ell},1/2\right)$ 
\STATE Graph $H \gets \edgeflip_{\po}(G)$
\IF {there exists $u \in S_a$ such that $N_H(u, S_b) > N_H(u, S_a)$}
\RETURN $\infty$ 
\ELSE 
\RETURN $N_H(S_1, S_2)$
\ENDIF
\caption{$\mathrm{VerifyDirected}(G, S_1, S_2, \eps, \delta)$ -- Verifier algorithm
		\label{algo:score}}
\end{algorithmic}
\end{algorithm}

By the disjoint-star degree counting mechanism the following lemma is immediate.
\begin{theorem}
  \label{thm:boosting:private}
For any graph input, $\eps, \delta > 0$, Algorithm $\mathrm{VerifyDirected}(G, S_1, S_2, \eps, \delta)$ is $(\eps,\delta)$-DP with respect to the input graph $G$. 
\end{theorem}
\begin{proof}
The algorithm's output is obtained by post-processing the output of a disjoint-star degree counting mechanism which is DP due to Lemma~\ref{def:disjoint-star-ada} and the (non-private) inputs $S_1,S_2$. 
\end{proof}

\subsection{Correctness}
We prove the correctness property of the private verification procedure required by the amplification framework. 
We use a series of facts about DSBMs that follow by the definition of the model. 
We write $B(n,p)$ to denote the Binomial distribution, which counts the number of 1s in $n$ independent trials of $\bern{p}$ random variables.

\begin{fact}
Consider a graph G in $DSBM(V_1, V_2, p,q)$, and any partitioning $(S_1, S_2) \neq (V_1, V_2$) such that $|S_1
| =|S_2|$. On each side of ($S_1$, $S_2$), 
 there exists at least one node $u$, such that $u$ has at least $n/4$ nodes from the same community of $u$ on the opposite side of the partition. 
\end{fact}

\begin{fact}\label{fact:n2q}
Consider a graph G in $DSBM(V_1, V_2, p,q)$. Let $\{a,b\} = \{1,2\}$, for every node $u \in V_a$, $|N_G(u, V_b)|$  is distributed as $B(n/2, q)$.
\end{fact}

\begin{fact}\label{fact26}
Consider a graph G in $DSBM(V_1, V_2, p,q)$ and a partitioning $(S_1, S_2) \neq (V_1, V_2$) such that $|S_1
| =|S_2|$. Let $\{a,b\} = \{1,2\}$. For a node  $u \in S_a$ with at least $n/4$ nodes of the same community of $u$ in $S_b$, $n_G(u, S_b)$ stochastically dominates the distribution of $B(n/4, q) + B(n/4, p)$.
\end{fact}

\begin{lemma}
\label{lem:onebadvertex}
    Consider a graph G in $DSBM(V_1, V_2, p,q)$ such that $\frac{p-q}{\sqrt{p}} \ge c \frac{\sqrt {\log(n)}}{\sqrt{n}}$ for a large enough but constant $c$ and fix one incorrect partition $S_1,S_2 \neq V_1,V_2$ with $|S_1|=|S_2|$. Assuming
    $|V_1 \cap S_1| \ge |V_1 \cap S_2|$ 
    and $|V_1 \cap S_2| \le n/8$, with probability at least
    $1-O(\exp(-c^2 |V_1 \cap S_2|\log n/3072))$, there exists a vertex 
    $u \in V_1 \cap S_2$ such that
    $n_G(u, S_1) > n_G(u, S_2)$.
\end{lemma}
\begin{proof}
    Let $Q := V_1 \cap S_2$ and since
    $|V_1 \cap S_1| \ge |V_1 \cap S_2|$ but $S_1,S_2 \neq V_1,V_2$, we have that $Q \neq \emptyset$.

    Our proof strategy is as follows: for each vertex $u$, let $\delta_u := n_G(u, S_1) - n_G(u, S_2)$. We will then show that with the prescribed 
    probability, $\sum_{u \in Q} \delta_u > 0$, which implies that there exists 
    at least one vertex $u^*\in Q$ for which $n_G(u^*, S_1) - n_G(u^*, S_2) > 0$.

    To do so, we will analyze $E[\sum_{u \in Q} n_G(u, S_1)]$ and 
    $E[\sum_{u \in Q} n_G(u, S_2)]$.
    
    For each vertex $u \in Q$, we count the expected number of neighbors in $S_1$.
    We have that $E[n_G(u, S_1)] = (|S_1| - |Q|) p + |Q|q$ 
    since $|S_1| = |S_2|$ and $|Q| = |V_1 \cap S_2| = |V_2 \cap S_1|$.
    Thus, $E[\sum_{u \in Q} n_G(u, S_1)] = |Q|((|S_1| - |Q|) p  +  |Q| q)$.

    Similarly, we have that $E[\sum_{u \in Q} n_G(u, S_2)]
    = |Q|((|Q|-1)p + (|S_2|-|Q|)q)\leq|Q|(|Q|p + (|S_2|-|Q|)q) $.

    Now we apply a standard Chernoff to both expectations to get an upper bound on $\sum_{u \in Q} n_G(u, S_2)$
    and a lower bound on $\sum_{u \in q} n_G(u, S_1)$. In particular, 
    
    \begin{align*}
        \Pr\Big[|\sum_{u \in Q} n_G(u, S_1) & - E[\sum_{u \in Q} n_G(u, S_1)]|  > (|S_1| - 2|Q|)|Q|(p - q)/8\Big]\\ &\leq 2 \exp \left(-\frac{(|S_1| - 2|Q|)^2|Q|^2(p - q)^2}{192E[\sum_{u \in Q} n_G(u, S_1)]}\right)\\
        & \leq 2 \exp \left(-\frac{(|S_1| - 2|Q|)^2|Q|^2(p - q)^2}{192 (|Q|(|S_1|p -|Q|(p - q)))}\right)\\
        & \leq 2 \exp \left(-\frac{(|S_1| - 2|Q|)^2|Q|(p - q)^2}{192 (|S_1|p)}\right)\\
        & \leq 2 \exp \left(-c^2\frac{(|S_1| - 2|Q|)^2|Q|\log n}{192|S_1|n}\right)\\
        & \leq 2 \exp \left(-c^2\frac{|Q|\log n}{3072}\right),
    \end{align*}
    where the last inequality follows
    from assuming that $|Q| \le n/8$.
    Similarly, we can bound the difference between  $\sum_{u \in Q} n_G(u, S_2)$ and  $E[\sum_{u \in Q} n_G(u, S_2)]$.

    We therefore have that with probability $1- O \left(\exp \left(-c^2\frac{|Q|\log n}{3072}\right)\right)$

    \begin{align*}
        \sum_{u \in Q} (n_G(u, S_1) - n_G(u, S_2)) &\ge 
    E[\sum_{u \in Q} n_G(u, S_1)] -     E[\sum_{u \in Q} n_G(u, S_2)] - 2\frac{(|S_1| - 2|Q|)|Q|(p - q)}{8} \\
    &\ge (|S_1| - 2|Q|)|Q|(p - q)/2 > 0,
        \end{align*}
    and the lemma follows.
\end{proof}

\begin{lemma}
\label{every-bad-partition-fails}
Consider a graph G in $DSBM(V_1, V_2, p,q)$ such that $\frac{p-q}{\sqrt{p}} \ge c \frac{\sqrt {\log(n)}}{\sqrt{n}}$. Then for a 
large enough constant $c$, with probability at least $1-n^{-O(1)}$
we have that for any incorrect partition $S_1,S_2 \neq V_1,V_2$ with $|S_1|=|S_2|$  and $|V_1 \cap S_2| \le n/8$ there exists a node $u$ in $S_1$ such that $n_G(u, S_2) - n_G(u, S_1) > 0$ or a node $v$ in $S_2$ such that $n_G(v, S_1) - n_G(v, S_2) > 0$.
\end{lemma}
\begin{proof}
The proof follows almost immediately from Lemma~\ref{lem:onebadvertex}. Indeed, observe that the number
of partitions of the vertex set into $S_1,S_2$ such that 
$|S_1|=|S_2|$ and such that $q := |S_1 \cap V_2| \le |S_1 \cap V_1|$
is at most $n \choose 2q$: they are obtained from the correct
partition and swapping a set of $q$ vertices from each side to the other. Hence, there are at most $n^{2q}$ such partition, which is
less than $\exp(c q \log n /3)$ for a large enough constant $c$.

Thus applying Lemma~\ref{lem:onebadvertex} for all $n/8$ values of $|Q| \ge 1$ 
we have that our statement holds with probability at least
$1-\exp(-c^2 \log n/3 -1)$.
\end{proof}

We now show the following lemma to bound the cost
of incorrect partition with a large deviation.
\begin{lemma}
\label{lem:concentrationcost}
    Consider a graph G in $DSBM(V_1, V_2, p,q)$ such that $\frac{p-q}{\sqrt{p}} \ge c \frac{\sqrt {\log(n)}}{\sqrt{n}}$ and $p \in \Omega\left( \frac{\log(2/\delta)}{\eps^2 n}\right)$. Then for a 
large enough constant $c$, with probability at least $1-n^{-O(1)}$
we have that for any incorrect partition $S_1,S_2 \neq V_1,V_2$ with $|S_1|=|S_2|$  and $|V_1 \cap S_2| \ge n/8$, we have that
$n_G(S_1, S_2) > n_G(V_1, V_2)$.
\end{lemma}
\begin{proof}
    This follows naturally from standard concentration bounds and a direct application of the Chernoff bound. 
    Indeed, let us first 
    bound the value of $n_G(V_1,V_2)$. 
    We have $E[n_G(V_1,V_2)] = q n^2 / 4$.
    Thus, we have that
    \begin{align*}
        \Pr[n_G(V_1,V_2) - E[n_G(V_1,V_2)] \geq 
        n^2 (p-q) / 32] &\le 2\exp\left(- \left(\frac{n^2 (p-q) / 32}{q n^2 / 4}\right)^2 \cdot q n^2 / 12\right)\\
        &\le 2\exp\left(- \frac{(p-q)^2}{64q } \cdot \frac{ n^2}{12 }\right)\\
        &\le 2\exp\left(- \frac{c^2}{768} n\log n\right)
        \end{align*}
        
    where the last inequality follows from our
    assumption on $(p-q)/p$. 

    We now apply a similar reasoning to any
    any incorrect partition $S_1,S_2 \neq V_1,V_2$ with $|S_1|=|S_2|$  and $|Q| := |V_1 \cap S_2| \ge n/8$.
    There we have that 
    $E[n_G(S_1,S_2)] = |Q| (n/2 - |Q|) p + |Q|^2q + (\frac{n}{2}- |Q|)^2q$ where 
    $|Q| \in [n/8, n/4]$ by our assumption. By noticing that the expectation is 
    minimized for $|Q|=n/8$ we have $E[n_G(S_1,S_2)]\geq \frac{3}{32}n^2p + \frac{5}{32}n^2q$.
    Similarly by applying Chernoff bound again we get $\Pr[b_G(V_1,V_2) - E[b_G(V_1,V_2)] \geq 
    n^2 (p-q) / 32] \leq 2\exp\left(- \frac{c^2}{768} n\log n\right)$. But this implies that with 
    probability $1-4\exp\left(- \frac{c^2}{768} n\log n\right)$, $n_G(S_1, S_2) > n_G(V_1, V_2)$ because
    the difference between their expectation is $\frac{3}{32}n^2(p-q)$.

    Since the number of such cut is less than $2^n$, 
    a union bound over all such cuts implies 
    that the lemma holds for all cuts simultaneously with
    high probability as desired.
\end{proof}

\begin{theorem}
\label{thm:verifier}
  Fix $\eps, \delta>0$, let $G$ be a graph from $DSBM(V_1, V_2, p,q)$. Suppose that $\frac{p-q}{\sqrt{p}} \in \Omega  
 \left(\frac{\sqrt{\log n}}{\sqrt{n}}\right)$ and $p \in \Omega\left( \frac{\log(2/\delta)}{\eps^2 n}\right)$, and let $n \geq \frac{390\log(2/\delta)}{\eps^2 }$. There is a $\theta$ such that 
$$\Pr_{G \sim DSBM(V_1, V_2, p,q)} \{\mathrm{VerifyDirected}(G, A, B, \eps, \delta) \le \theta \}  \ge 1-n^{-O(c^2)}$$ if $(A,B) = (V_1,V_2)$
 and for any $(A,B)\neq (V_1,V_2)$, 
 $$\Pr_{G \sim DSBM(V_1, V_2, p,q)} \{\mathrm{VerifyDirected}(G, A, B, \eps, \delta)> \theta \}\ge 1-n^{-O(c^2)}.$$ 
\end{theorem}
\begin{proof}
First we prove, similarly to Theorem~\ref{thm:main-accuracy-dsbm}, that the flipped graph comes from a DSBM distribution with good parameters.
Algorithm~\ref{algo:score}, in the edge-flipping process, sets $\ell = \frac{n}{2}$ and $\po = \min\left(\frac{96 \log(2/\delta)}{\eps^2 \ell},1/2\right) = \min\left(\frac{192 \log(2/\delta)}{\eps^2 n},1/2\right)$. By assumption,   $n \geq \frac{390\log(2/\delta)}{\eps^2 }$ hence $\frac{192 \log(2/\delta)}{\eps^2 n} = c < 1/2$, (for some constant $c$). So $\po = \frac{192 \log(2/\delta)}{\eps^2 n}$ and $1-2\po \geq 1-2c \in \Omega(1)$. This also implies that $p \in \Omega(\po)$ by the assumption on $p$, and so $p(1-2\po) + \po \in O(p)$ The argument follows like in Theorem~\ref{thm:main-accuracy-dsbm} to conclude that graph $H$ comes from  $DSBM(V_1,V_2,p',q')$ with $\frac{p'-q'}{\sqrt{p'}} \in \Omega\left(\frac{\sqrt{\log n}}{\sqrt{n}}\right)$.

Finally, all we have to do is apply Lemmas~\ref{every-bad-partition-fails} and~\ref{lem:concentrationcost}. Indeed,
we appeal to Lemma~\ref{every-bad-partition-fails} to show that any incorrect partition that would misclassify less than $n/8$ vertices would contain a vertex that has less neighbor in the part it is assigned to than in the other part, and in which case the Verifier will return $\infty$. Moreover, Lemma~\ref{lem:concentrationcost} shows that the planted cut has fewer edges in the cut than any cut that misclassifies more than $n/8$ vertices
and so there is a $\theta$ such that for the correct partition the Verifier outputs a score lower than $\theta$ (the number of edges in the cut) and above $\theta$ for all incorrect partitions (either $\infty$ or the number of edges in the cut).
\end{proof}

\section{Experimental Results for Directed SBM Graphs}\label{app-experiment-sbm}

In this section, we report a preliminary empirical evaluation of our algorithms for directed Stochastic Block Model estimation. We first describe the experimental setup, followed by a discussion on graph size and its impact on our algorithm relative to established baselines. We then detail the modifications made to improve the algorithm's practical performance while retaining provable privacy guarantees. Finally, we present and analyze our comparative results.

\subsection{Experimental Setup}
We implemented all evaluated algorithms in Python and conducted experiments on a standard single-core machine with 100GB of RAM and x64 architecture. The source code for these experiments will be made available as open source upon publication. All experiments were performed on synthetic graphs generated from the SBM model and other random graph models (see the following subsection).

\subsection{Impact of Graph Size on $\po$}
A primary consideration in evaluating our algorithm is the impact of graph size on the perturbation probability required for privacy. Recall that our DSBM algorithm perturbs the graph with probability:
$$\po = \min\left(\frac{96 \log(2/\delta)}{\eps^2 \ell},1/2\right) = \min\left(\frac{1728 \sqrt{\log n}\log(2/\delta)}{\eps^2 n},1/2\right).$$
Notably, this probability decreases as $n$ increases, meaning our algorithm becomes increasingly precise as the graph grows. Conversely, for small $n$, this probability may be capped at $1/2$, resulting in a complete loss of signal. Consequently, our algorithm yields non-trivial results only on graphs exceeding a minimum size (determined by the privacy requirements). We consider this the main practical limitation of our work, which we partially address by refining the computation of $\po$ through numerical analysis (see Section~\ref{app-modifications}).

With this modification, we observed that the algorithm is viable for graphs with approximately 10,000 nodes at strong privacy settings, with solution quality improving consistently as $n$ increases beyond this threshold.

\subsection{Competing Baselines}
Our goal is to compare our algorithm with known differentially private baselines for SBM graphs. As noted, the limitations regarding small graphs prevent us from running our method on graphs with fewer than 10,000 nodes. Therefore, we cannot provide a direct comparison with the advanced numerical solver-based algorithms in \cite{seif2022differentially} and \cite{he2024differentially}, which were evaluated on graphs with at most 200 nodes. For graphs of that scale, their approach is preferable.

Instead, we evaluate our algorithm against two more efficient DP baselines. First, we consider the graph perturbation method from~\cite{seif2022differentially}, which produces a DP graph followed by a standard NumPy second-eigenvector computation on a symmetrized version of the graph. This provides a strong baseline, albeit at a higher computational cost (referred to as \textsl{Basic+Spectral}). Second, we use the graph perturbation method from~\cite{seif2022differentially} followed by the non-private algorithm from \cite{cohen2022community} (referred to as \textsl{Basic+Cohen-Addad}). Note that these baselines are $\eps$-DP, whereas our algorithm is $(\eps, \delta)$-DP; thus, the comparison is not strictly ``apples-to-apples''.

\subsection{Practical Modifications to the Framework}
\label{app-modifications}
A critical component of our method is the calculation of the probability of perturbation $\po$ required to ensure privacy. Although the formula above provides a sufficient theoretical upper bound, we investigated whether a lower $\po$ could be utilized. In our implementation, we use numerical approximation to find a $\po$ that is provably sufficient for privacy.

Specifically, for a given lower bound $\ell$ on the minimum size required to compute a degree, we find a $\po(\ell)$ that satisfies the conditions in the proof of Lemma~\ref{lemma:dp-degree}. This is achieved via a binary search on $\po$ until a value is found where the distributions of $B'_{plus}$ and $B_{plus} + 1$ meet the $(\eps, \delta)$ requirements for sets of size $\ell$. This computation requires only $\tilde O(\ell)=\tilde O(n)$ time (for trial of a probability) and yields a significantly lower probability than the theoretical upper bound while maintaining identical privacy guarantees. We use binary search with an approximation of $1\%$ for the minimum probability found and always use a probability which yields the DP requirement. The resulting improvements are substantial, as illustrated in Figure~\ref{fig:p}.

\begin{figure}[h]
    \centering
    \includegraphics[width=0.6\textwidth]{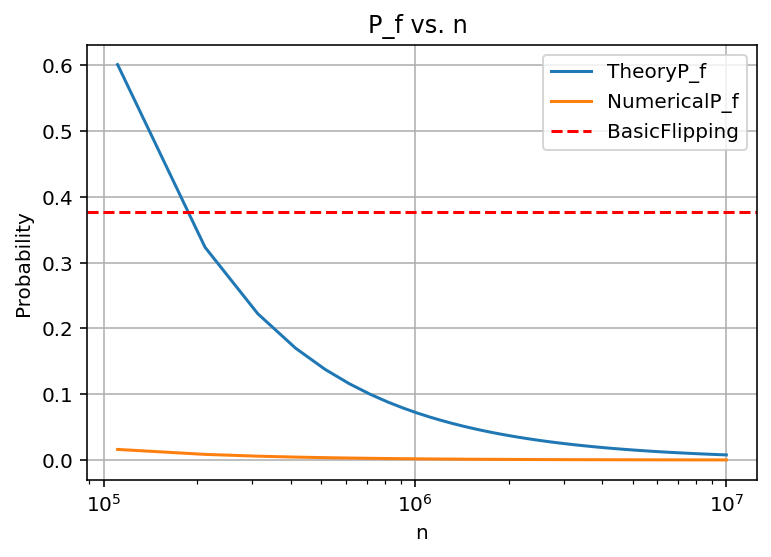}
    \caption{Perturbation probability required for our algorithm based on the theoretical upper bound (TheoryP\_f), numerical analysis (NumericalP\_f), and the standard $\eps$-DP graph perturbation \cite{seif2022differentially} (BasicFlipping). Results assume $\eps=0.5$, $\delta=5\cdot10^{-5}$ (for our approximate DP methods), and the minimum size $\ell$ required by our directed algorithm.}
    \label{fig:p}
\end{figure}

The numerical method allows for a significantly lower perturbation probability at small $n$ values, eventually converging to the theoretical bound as $n$ increases. Furthermore, in this regime, the numerically computed probability is much lower than the constant perturbation probability required for $\eps$-DP by the basic flipping method. We also introduced a non-privacy-affecting modification to increase efficiency: we run the algorithm only once and omit the boosting phase.

\subsection{Results}
We now report the comparison between our method and the two chosen baselines. We generated DSBM graphs of increasing size $n$ with fixed probabilities $p=0.1$ and $q=0.07$, using $\eps=0.5$ and $\delta=10^{-5}$. The results are detailed in Table~\ref{tab:results}.

\begin{table}[htbp]
\centering
\caption{Comparison of Accuracy and Running Time for Different $n$\label{tab:results}}
\begin{tabular}{ccccccc}
\toprule
$n$ & \multicolumn{2}{c}{Ours} & \multicolumn{2}{c}{\textsl{Basic + Cohen-Addad}} & \multicolumn{2}{c}{\textsl{Basic + Spectral}} \\
\cmidrule(lr){2-3} \cmidrule(lr){4-5} \cmidrule(lr){6-7}
& Accuracy & Time (sec) & Accuracy & Time (sec) & Accuracy & Time (sec) \\
\midrule
10,000 & 0.6671 & 50.52  & 0.6358 & 88.20  & 0.9864 & 355.57  \\
15,000 & 0.9809 & 110.45 & 0.8982 & 173.25 & 0.9977 & 1139.35 \\
20,000 & 0.9979 & 294.65 & 0.9489 & 287.44 & 0.9994 & 2736.65 \\
25,000 & 0.9999 & 397.17 & 0.9664 & 459.87 & -- & -- \\
30,000 & 1.0000 & 540.82 & -- & -- & -- & -- \\
40,000 & 1.0000 & 908.65 & -- & --     & --     & --      \\
\bottomrule
\end{tabular}
\end{table}

Several observations are noteworthy. First, our algorithm's running time is significantly lower than that of the spectral-based algorithm and lower or comparable to the \textsl{Basic+Cohen-Addad} baseline. Notably, the baselines could not process graphs with 40,000 nodes within a one-hour time limit or without exhausting available memory. In contrast, our algorithm scales to graphs two orders of magnitude larger than prior work utilizing numerical solvers (which where ran only on graphs with 200 nodes~\cite{he2024differentially, seif2022differentially}).

In terms of accuracy (reported as the ratio of correctly classified nodes), our algorithm becomes increasingly precise as $n$ grows, achieving $100\%$ accuracy at $n=30,000$. Our method consistently outperforms the \textsl{Basic + Cohen-Addad} baseline. While the \textsl{Basic + Spectral} baseline is highly accurate where it can scale, our results confirm our theoretical claims: our framework enables exact recovery for significantly larger graphs than previously possible.

\section{Experimental Results on the Use of Our Framework for Private Degree Counting}\label{app-experiment-degree} 

In this section, we demonstrate an additional application of our degree-counting framework to problems beyond SBM community detection: outputting a graph’s degree sequence with edge differential privacy. More precisely, given a graph $G=(V,E)$, the problem seeks to output an approximate degree sequence $(d'_1, \ldots, d'_n)$ (with edge-level DP) that has the highest possible correlation with the correct sequence $(\deg(v_1), ...,  \deg(v_n))$ of the degrees of nodes in $G$. 
This is a fundamental graph problem that was previously studied in~\cite{degree_estimation} with differential privacy. (We note however that our framework provides $(\eps,\delta)$-DP guarantees while~\cite{degree_estimation} focused on $\eps$-DP, so the results are not immediately comparable.)

We provide an empirical evaluation of our framework for this simple but foundational problem in graphs.  Our framework provides an immediate solution to the degree estimation problem.
For directed graphs, it is sufficient to compute the disjoint-star queries of $(u, V-{u})$ for all nodes $u$. For undirected graphs, we apply the same strategy from Section~\ref{sect:SBM} of accounting for the effect of $u,v$ on two undirected queries.
This results in an easy to implement algorithm that runs in  near-linear time: simply perturb the edges with the probability given by ~\ref{lemma:dp-degree} (setting $\ell = n-1$) and use the non-private exact degree counting algorithm on the perturbed graph. 
Notice that this application allows the lowest possible perturbation noise from our framework as all degrees are computed to the largest possible set $V$ and with a single parallel computation (in directed graphs). (See Lemma~\ref{lemma:dp-degree} for the tradeoff between size and noise.)

For this experiment, to obtain more interesting results, we used a random graph with a heavy-tailed degree sequence (rather than the SBM model, which has a concentrated degree distribution). We used the standard Barabasi-Albert model from the networkx library with parameters $n = 50,000, m = 500$. This results in a power law degree distribution for the degrees of the nodes in a graph with a total of $50,000$ nodes and $24,750,000$ edges. Then we ran our perturbation framework (with $\eps$ = 4, $\delta$ = 1e-5) with min size $\ell=(|V|-1$). Our algorithm perturbed the graph within \emph{250 seconds}. 
Then, we applied the simple non-private degree counting algorithm on this perturbed graph (which takes linear time) and compared the results with obtaining the exact  degree sequence on the original graph. For each node $u$, we compared the degrees of $u$ in the original graph and the perturbed graph (which is DP) and computed the correlation coefficients of the two values. We observe a strong correlation between the ground truth and the DP estimate: Pearson Correlation 0.999  (p-value 0); Spearman Rank Correlation 0.994 (p-value 0).  This confirms the accuracy of our framework for this simple task. 
\section*{Conclusion and Future Work}
We present an efficient and simple algorithm to privately detect communities in the stochastic block model. A natural question for future work is improving our algorithm to get tighter bounds. Additionally, we note that our general framework can easily be applied to more general model settings whenever a suitable non-private algorithm for the problem is available.
Thus extending the results to generalized SBMs with more or uneven communities using a modified algorithm is a natural next step.
Finally, it would be interesting to understand which other classical algorithms have variants that can be made DP using our privatization scheme from Section \ref{sec:edgeflip}. Natural candidate problems in this direction are densest subgraphs and planted clique.

\section*{Acknowledgements}
Hanna Koml\'os was supported in part by the Graduate Fellowships for STEM Diversity. This work was done in part while Hanna Koml\'os was visiting the Simons Institute for the Theory of Computing, and was supported in part by Google Research, Fall 2025.

\bibliographystyle{unsrtnat}
\bibliography{references}

\ifappendix

\appendix

\begin{appendix}

\section{Proof of \Cref{lemma:dp-long-calcs}}\label{app:proof-longcalcs}

\begingroup
\def\thetheorem{\ref{lemma:dp-long-calcs}}
\begin{claim}
For all sets $F\subseteq \mathbb{N}$,
$\Pr[n_H(u,S) \in F] \leq e^{\eps}  \Pr[n_{H'}(u,S) \in F] + \delta$ and 
$\Pr[n_{H'}(u,S) \in F] \leq e^{\eps}  \Pr[n_{H}(u,S) \in F] + \delta.$
\end{claim}
\endgroup

\begin{proof}
Recall that $n_H(u,S) = x - B_{minus} + B_0 + B_{plus}$ and $n_{H'}(u,S) = x + 1 -  B'_{minus}  - B'_0 +B'_{plus}$.
We first argue that if $\po = 1/2$, the random variables for $n_H(u,S)$ and $n_{H'}(u,S)$ are identically distributed. Indeed, since $B_{minus} \sim \bin{x}{1/2}$, $B_{plus} \sim \bin{|S| - x-1}{1/2}$, and $B_0 \sim \bern{1/2} = \bin{1}{1/2}$, we have \begin{align*} n_H(u,S) &\sim x - \bin{x}{1/2} + \bin{1}{1/2} + \bin{|S| - x-1}{1/2}= x + \bin{|S|-x}{1/2} - \bin{x}{1/2} = \bin{|S|}{1/2}.\end{align*}
Similarly, $n_{H'}(u,S) \sim  x+1 + \bin{|S|-x-1}{1/2} - \bin{x+1}{1/2} = \bin{|S|}{1/2}$. 
We can therefore restrict from now on to the case that $\po = \frac{96\log(2/\delta)}{\ell\eps^2} \leq 1/2$.

The remainder of this proof is based on the proof of Theorem A.1 in \cite{bandits_shuffle_TKMS21}.

Suppose first that $x \leq |S|/2 - 1$.
We proceed by a coupling argument, fixing $B_{minus}=B'_{minus}=j$ and sampling the other component variables of $n_H(u,S)$ and $n_{H'}(u,S)$ independently.

Fix $k$. Showing that $\Pr(n_H(u,S) = k)$ is $(\eps,\delta)$-close to $\Pr(n_{H'}(u,S) = k)$ is equivalent to showing that $\Pr(B_{plus}+B_0= k - x + j)$ is $(\eps,\delta)$-close to $\Pr(B'_{plus}-B'_0 = k - x + j - 1)$. First we condition on the outcome $B_0=B'_0=1$. In this case we are comparing $\Pr(B_{plus}= k - x + j - 1)$ and $\Pr(B'_{plus} = k - x + j)$.

Write $t=k-x+j$ and $m=|S| - x - 1$,
and let
$\tau = E[B_{plus}] = m\po$.  Then by assumption we have $\tau \leq m/2$, and $m \geq |S| - (|S|/2-1) - 1 = |S|/2 \geq \ell/2$ and so $\tau = \frac{96 m\log(2/\delta)}{\ell\eps^2} \geq \frac{48\log(2/\delta)}{\eps^2}$.

Since $B_{plus}$ is a binomial random variable, we can apply a Chernoff bound such as in \citep{shuffle_distributed_CSUZZ19} to get $\Pr(|B_{plus}-E[B_{plus}]| \geq \sqrt{3\tau\log(2/\delta)}) < \delta$, so $\Pr(B_{plus} \in I_c) \geq 1-\delta$ 
for $I_c := ( \tau - \sqrt{3\tau\log(2/\delta)} , \tau + \sqrt{3\tau\log(2/\delta)} )$, and similarly for $B'_{plus}$. 

Let $t \in I_c$. Then
\begin{align*}
\frac{ \Pr(B'_{plus}=t) }{ \Pr(B_{plus}=t-1) } &= \frac{{m \choose t}}{{m \choose t-1}} \cdot \frac{\po^t(1-\po)^{m-t}}{\po^{t-1}(1-\po)^{m-t+1}} \tag{1}\label{eq1}
\\ & = \frac{m-t+1}{t} \cdot \frac{\po}{1-\po} 
\\ & \leq \frac{m - \tau + \sqrt{3\tau\log(2/\delta)}+1}{\tau - \sqrt{3\tau\log(2/\delta)}} \cdot \frac{\tau/m}{1-\tau/m} \label{eq2}\tag{2}
\\ & =  \frac{m - \tau + \sqrt{3\tau\log(2/\delta)}+1}{m-\tau} \cdot \frac{\tau}{\tau - \sqrt{3\tau\log(2/\delta)}}
\\ & = \left( 1 +  \frac{\sqrt{3\tau\log(2/\delta)}+1}{m-\tau} \right) \frac{1}{1 - \sqrt{3\log(2/\delta)/\tau}}
\\ & \leq \left( 1 +  \frac{\sqrt{3\tau\log(2/\delta)}+1}{\tau} \right) \frac{1}{1 - \sqrt{3\log(2/\delta)/\tau}} \label{eq3}\tag{3}
\\ & = \frac{1+\sqrt{3\log(2/\delta)/\tau}+1/\tau}{1 - \sqrt{3\log(2/\delta)/\tau}}
\\ & \leq \frac{1+\eps/4+1/\tau}{1-\eps/4} \label{eq4}\tag{4}
\\ & \leq \frac{1+\eps/2}{1-\eps/4} \label{eq5}\tag{5} 
\\ & \leq e^{\eps}. \label{eq6}\tag{6}
\end{align*}

Equation~(\ref{eq1}) holds by definition of the binomial distribution. 
Equation~(\ref{eq2}) holds by substituting $t \geq \tau - \sqrt{3\tau\log(2/\delta)}$, which follows since $t \in I_c$. 
Equation~(\ref{eq3}) holds since $\tau \leq m/2$ implies $m-\tau \geq \tau$.
Equation~(\ref{eq4}) holds since $m\geq \ell/2$ implies $\frac{3\log(2/\delta)}{\tau} = \frac{3\log(2/\delta)}{m\po} = \frac{3\ell\eps^2\log(2/\delta)}{96 m\log(2/\delta)} \leq \frac{\eps^2}{16}.$
Equation~(\ref{eq5}) holds since $\tau \geq \frac{48\log(2/\delta)}{\eps^2} \geq \frac{4}{\eps}$, because we assume $\eps \leq 12 \log (2/\delta)$.
Finally, Equation~(\ref{eq6}) holds since $\frac{1+z/2}{1-z/4} \leq e^z$ for all $z \in [0,1]$.

To show that $\frac{ \Pr(B'_{plus}=t) }{ \Pr(B_{plus}=t-1) } \geq e^{-\eps}$, we substitute $t \leq \tau + \sqrt{3\tau\log(2/\delta)}$ in Equation~(\ref{eq2}) instead: 

\begin{align*}
\frac{ \Pr(B'_{plus}=t) }{ \Pr(B_{plus}=t-1) } 
& = \frac{m-t+1}{t} \cdot \frac{\po}{1-\po} 
\\ & \geq \frac{m - \tau - \sqrt{3\tau\log(2/\delta)}+1}{\tau + \sqrt{3\tau\log(2/\delta)}} \cdot \frac{\tau/m}{1-\tau/m} 
\\ & =  \frac{m - \tau - \sqrt{3\tau\log(2/\delta)}+1}{m-\tau} \cdot \frac{\tau}{\tau + \sqrt{3\tau\log(2/\delta)}}
\\ & = \left( 1 -  \frac{\sqrt{3\tau\log(2/\delta)}-1}{m-\tau} \right) \frac{1}{1 + \sqrt{3\log(2/\delta)/\tau}}
\\ & \geq \left( 1 -  \frac{\sqrt{3\tau\log(2/\delta)}}{m/2} \right) \frac{1}{1 + \eps/4} \label{eq7}\tag{7}
\\ & = \frac{1-2\sqrt{3\tau\log(2/\delta)/m^2}}{1 + \eps/4}
\\ & \geq \frac{1-\eps/4}{1+\eps/4} \label{eq8}\tag{8}
\\ & \geq e^{-\eps}. \label{eq9}\tag{9}
\end{align*}

Equation~(\ref{eq7}) holds since $\tau \leq m/2$ and $\frac{3\log(2/\delta)}{\tau} \leq \frac{\eps^2}{16}$ as above.
Equation~(\ref{eq8}) holds since $m\geq \ell/2$ and $\frac{96\log(2/\delta)}{\ell\eps^2} \leq \frac{1}{2}$ imply $$\frac{3\tau\log(2/\delta)}{m^2} \leq \frac{3\log(2/\delta)}{2m} \leq \frac{3\log(2/\delta)}{\ell} \leq \frac{\eps^2}{64}.$$
Equation~(\ref{eq9}) holds since $\frac{1-z/4}{1+z/4} \geq e^{-z}$ for all $z \in [0,1]$.
Thus we have shown that $e^{-\eps} \leq \frac{\Pr(B'_{plus}=t)}{\Pr(B_{plus}=t-1)} \leq e^{\eps}$ for all $t \in I_c$.
By construction, this is equivalent to showing that, $e^{-\eps} \leq \frac{\Pr(n_{H'}(S,x))=k)}{\Pr(n_{H}(S,x))=k)} \leq e^{\eps}$ for $k-x+j \in I_c$.

\smallskip
Let $F \subseteq \mathbb{N}$. Then
\begin{align*}
\Pr[n_{H'}(u,S) \in F] 
 & = \Pr[n_{H'}(u,S) \in F \wedge B'_{plus} \in I_c] + \Pr[n_{H'}(u,S) \in F \wedge B'_{plus} \notin I_c]
\\ & \leq \Pr[n_{H'}(u,S) \in F \wedge B'_{plus} \in I_c] + \delta
\\ & = \delta +\sum_{k \in F} \Pr[n_{H'}(u,S) = k \wedge B'_{plus} \in I_c]
\\ & = \delta +  e^{\eps} \sum_{k \in F}\Pr[ n_{H}(u,S) = k \wedge B'_{plus} \in I_c]
\\ & \leq \delta +  e^{\eps} \sum_{k \in F}\Pr[ n_{H}(u,S) = k]
\\ & = \delta + e^{\eps}\Pr[n_{H}(u,S) \in F].
\end{align*}

A dual argument shows that $\Pr[n_{H}(u,S) \in F] \leq e^{\eps}\Pr[n_{H'}(u,S) \in F] + \delta$.

The case that $B_0=B'_0=0$ is the same as the above argument with the roles of $B_{plus}$ and $B'_{plus}$ reversed. The remaining cases are trivial, since  ($B_0=0$, $B'_0=1$) is equivalent to the condition $B_{plus}=k-x-j=B'_{plus}$ 
($B_0=1$, $B'_0=0$), is equivalent to the condition $B_{plus}=k-x-j-1=B'_{plus}$, and $B_{plus}$ and $B'_{plus}$ are identically distributed.

Now suppose $x \geq |S|/2$. In this case, we instead fix $B_{plus}=B'_{plus}=j$, and repeat the above argument on $B_{minus}$ and $B'_{minus}$ with $t=x+j-k$ and $\tau = x\pl$. Since $B_{minus}, B'_{minus}$ are identically distributed as $\bin{x}{\pl}$ with at least $|S|/2$ trials, the logic is the same.
\end{proof}

\section{Efficient implementation of the $\edgeflip_{\po}$ mechanism}\label{app:algo-flip}

Here we present the pseudocode for the efficient implementation of the $\edgeflip_{\po}$ mechanism. We write $\bern(p)$ to denote the Bernoulli distribution, which takes value 1 with probability $p$ and 0 with probability $1-p$. 

\begin{algorithm}

\noindent \textbf{Input:} A directed graph $G = (V, E)$, edge flipping probability $\po$.

\noindent \textbf{Output:} A directed graph $H = (V,E')$ where the edges/non-edges of $G$ are flipped independently

\begin{algorithmic}[1]
\STATE 
  $E' \gets \emptyset$

\FOR {$(u,v) \in E$} 
    \IF {$\bern{1-\pl} = 0$}
       \STATE $E' \gets E' \cup \{(u,v)\}$
    \ENDIF
\ENDFOR
\IF {$|E| > \frac{1}{100}|V|^2 \vee \po > \frac{1}{100}$}
    \FOR {$(u,v) \notin E$} 
        \IF {$\bern{\po} = 1$}
           \STATE $E' \gets E' \cup \{(u,v)\}$
        \ENDIF
    \ENDFOR
\ELSE
    \STATE Draw $m_{+} = \bin{n^2- |E|}{\po}$  
    \STATE Add to $E'$, $m_{+}$ edges sampled without replacement from $V \times V - E$.
\ENDIF
\RETURN Graph $H \gets (V,E')$
\caption{$\mathrm{EdgeFlipping}(G,\po)$ --- Efficient implementation of the $\edgeflip_{\po}$ mechanism}
\label{alg:flip}
\end{algorithmic}
\end{algorithm}
\section{Pseudocode for subroutines used in main algorithms}\label{app:subroutines}


\begin{algorithm}[H]
\noindent \textbf{Input:} A graph $G = (V, E)$ with $n$ vertices, minimum subset size $\ell$

\noindent \textbf{Output:} A bipartition of $V$
\begin{algorithmic}[1]
\STATE $b \gets \sqrt{\log n}$
\IF{$b$ is even}
    \STATE $b \gets b+1$
\ENDIF
\STATE $S_1,S_2, \dots S_b \gets$ a random partition of $V$ into $b$ subsets of equal size
\FOR{$i = 1,\dots b$} 
    \STATE $S_{i,1},S_{i,2} \gets$ random partition of $S_i$ into 2 subsets of equal size
\ENDFOR
\STATE $k \gets \binom{b}{2}$
\STATE $\left(e_1 = (x_1,y_1), \dots , e_{k}=(x_{k},y_{k}))\right) \gets$ an Eulerian tour on $K_b$, the complete graph on $b$ vertices
\FOR{$i = 1 \dots k$}
    \STATE $S_{y_i,1}, S_{y_i,2} \gets \mathrm{Update}(G,\ell,S_{y_i},S_{x_i,1},S_{x_i,2})$
\ENDFOR
\STATE $S_1^* \gets \bigcup_{i\leq b} S_{i,1}$, $S_2^* \gets \bigcup_{i\leq b} S_{i,2}$
\RETURN $S_1^*, S_2^*$ 
\caption{$\mathrm{Phase 1}(G)$} 
\end{algorithmic}
\end{algorithm}


\begin{algorithm}
\noindent \textbf{Input:} Graph $G=(V,E)$, minimum subset size $\ell$, subsets $S_j,S_{i_1},S_{i_2} \subseteq V$

\noindent \textbf{Output:} A bipartition of $S_j$
\begin{algorithmic}[1]
\STATE $S_{j,1},S_{j,2} \gets \emptyset$
\IF{$|S_{i,1}| < \ell$}
    \STATE $T_2 \gets$ a random subset of $S_{i,2}$ of size $\ell - |S_{i,1}|$
    \STATE $S_{i,1} \gets S_{i,1} \cup T_2$
\ENDIF
\IF{$|S_{i,2}| < \ell$} 
    \STATE $T_1 \gets$ a random subset of $S_{i,1}$ of size $\ell - |S_{i,2}|$
    \STATE $S_{i,2} \gets S_{i,2} \cup T_1$
\ENDIF

\IF{$|S_{i,1}| > |S_{i,2}|$} 
    \STATE $\tilde{S}_{i,1} \gets$ a subset of $\tilde{S}_{i,1}$ of size $|S_{i,2}|$ chosen uniformly at random
    \STATE $\tilde{S}_{i,2} = S_{i,2}$
\ELSE
    \STATE $\tilde{S}_{i,2} \gets$ a subset of $\tilde{S}_{i,2}$ of size $|S_{i,1}|$ chosen uniformly at random
    \STATE $\tilde{S}_{i,1} = S_{i,1}$
\ENDIF

\FOR{$u \in S_j$} 
    \IF{$n_G(u,S_{i,1}) > n_G(u,S_{i,2})$}
        \STATE $S_{j,1} \gets S_{j,1} \cup \{u\}$
    \ELSE 
    \IF{$n_G(u,S_{i,2}) > n_G(u,S_{i,1})$}
        \STATE $S_{j,2} \gets S_{j,2} \cup \{u\}$
    \ELSE 
        \STATE Assign $u$ to $S_{j,1}$ or $S_{j,2}$ randomly
    \ENDIF
    \ENDIF
\ENDFOR
\RETURN $S_{j,1}, S_{j,2}$
\caption{$\mathrm{Update}(G,\ell,S_j,S_{i_1},S_{i_2})$} 
\end{algorithmic}
\end{algorithm}

\newpage
\section{Pseudocode for Intermediate DSBM Algorithm}\label{app:intermediate}
\begin{algorithm}[ht]
\noindent \textbf{Input:} A directed graph $G = (V, E)$ with $n$ vertices, privacy parameters $\eps,\delta$

\noindent \textbf{Output:} A bipartition of $V$

\begin{algorithmic}[1]
\STATE $\ell \gets \frac{n}{18\sqrt{\log n}}$
\STATE $\po \gets \min\left(\frac{96\log(2/\delta)}{\eps^2 \ell},1/2\right)$
\STATE Graph $H \gets \edgeflip_{\po}(G)$

\STATE $S,S' \gets$ a partition of $V$ randomly into two disjoint subsets of size $n/2$
\STATE $H_S \gets$ the subgraph of $H$ induced by $S$; $H_{S'} \gets$ the subgraph of $H$ induced by $S'$

\STATE $S_1,S_2 \gets \mathrm{Phase 1-Intermediate}(H_S)$
\STATE $S'_1,S'_2 \gets \mathrm{Phase 1-Intermediate}(H_{S'})$
\STATE $S^*_1,S^*_2 \gets \mathrm{Update-Intermediate}(H,S,S'_1,S'_2)$
\STATE $S^{'*}_1,S^{'*}_2 \gets \mathrm{Update-Intermediate}(H,S',S_1,S_2)$
\RETURN $S^*_1 \cup S^{'*}_1, S^*_2 \cup S^{'*}_2$

\caption{$\mathrm{PrivateDSBMRecovery-Intermediate}(G,\eps,\delta)$ an algorithm for exact recovery in the DSBM}\label{alg:dir_int}
\end{algorithmic}
\end{algorithm}

\begin{algorithm}
\noindent \textbf{Input:} Graph $G=(V,E)$, subsets $S_j,S_{i_1},S_{i_2} \subseteq V$

\noindent \textbf{Output:} A bipartition of $S_j$
\begin{algorithmic}[1]
\STATE $S_{j,1},S_{j,2} \gets \emptyset$

\IF{$|S_{i,1}| > |S_{i,2}|$} 
    \STATE $\tilde{S}_{i,1} \gets$ a subset of $\tilde{S}_{i,1}$ of size $|S_{i,2}|$ chosen uniformly at random
    \STATE $\tilde{S}_{i,2} = S_{i,2}$
\ELSE
    \STATE $\tilde{S}_{i,2} \gets$ a subset of $\tilde{S}_{i,2}$ of size $|S_{i,1}|$ chosen uniformly at random
    \STATE $\tilde{S}_{i,1} = S_{i,1}$
\ENDIF

\FOR{$u \in S_j$} 
    \IF{$n_G(u,S_{i,1}) > n_G(u,S_{i,2})$}
        \STATE $S_{j,1} \gets S_{j,1} \cup \{u\}$
    \ELSE 
    \IF{$n_G(u,S_{i,2}) > n_G(u,S_{i,1})$}
        \STATE $S_{j,2} \gets S_{j,2} \cup \{u\}$
    \ELSE 
        \STATE Assign $u$ to $S_{j,1}$ or $S_{j,2}$ randomly
    \ENDIF
    \ENDIF
\ENDFOR
\RETURN $S_{j,1}, S_{j,2}$
\caption{$\mathrm{Update-Intermediate}(G,S_j,S_{i,1},S_{i,2})$} 
\end{algorithmic}
\end{algorithm}


\begin{algorithm}[H]
\noindent \textbf{Input:} A graph $G = (V, E)$ with $n$ vertices

\noindent \textbf{Output:} A bipartition of $V$
\begin{algorithmic}[1]
\STATE $b \gets \sqrt{\log n}$
\STATE $S_1,S_2, \dots S_b \gets$ a random partition of $V$ into $b$ subsets of equal size
\FOR{i = 1,\dots b} 
    \STATE $S_{i,1},S_{i,2} \gets$ random partition of $S_i$ into 2 subsets of equal size
\ENDFOR
\STATE $k \gets \binom{b}{2}$
\STATE $\left(e_1 = (x_1,y_1), \dots , e_{k}=(x_{k},y_{k}))\right) \gets$ an Eulerian tour on $K_b$, the complete graph on $b$ vertices
\FOR{$i = 1 \dots k$}
    \STATE $S_{y_i,1}, S_{y_i,2} \gets \mathrm{Update}(G,S_{y_i},S_{x_i,1},S_{x_i,2})$
\ENDFOR
\STATE $S_1^* \gets \bigcup_{i\leq b} S_{i,1}$, $S_2^* \gets \bigcup_{i\leq b} S_{i,2}$
\RETURN $S_1^*, S_2^*$ 
\caption{$\mathrm{Phase 1-Intermediate}(G)$} 
\end{algorithmic}
\end{algorithm}

\fi

\end{appendix}

\end{document}